\documentclass[11pt]{article}
\usepackage{amssymb}

\usepackage{booktabs} 

\usepackage{multirow}

\usepackage{algorithm}
\usepackage{algorithmic}

\usepackage[table]{xcolor}

\usepackage{bm}

\newcommand{\defeq}{\coloneqq}

\usepackage{complexity}
\newcommand{\CLS}{\mathsf{CLS}}

\let\R\relax

\newcommand{\easybadge}{\colorbox{green!12}{\strut \textcolor{green!45!black}{$\mathsf{P}$}}}
\newcommand{\clsbadge}{\colorbox{yellow!18}{\strut \textcolor{orange!80!black}{$\mathsf{CLS}$-complete}}}
\newcommand{\ppadbadge}{\colorbox{red!10}{\strut \textcolor{red!70!black}{$\mathsf{PPAD}$-complete}}}

\usepackage{hyperref}

\newcommand{\declarecolor}[2]{\definecolor{#1}{RGB}{#2}\expandafter\newcommand\csname #1\endcsname[1]{\textcolor{#1}{##1}}}
\declarecolor{White}{255, 255, 255}
\declarecolor{Black}{0, 0, 0}
\declarecolor{Maroon}{128, 0, 0}
\declarecolor{Coral}{255, 127, 80}
\declarecolor{Red}{182, 21, 21}
\declarecolor{LimeGreen}{50, 205, 50}
\declarecolor{DarkGreen}{0, 80, 0}
\declarecolor{Purple}{146, 42, 158}
\declarecolor{Navy}{0, 0, 128}
\declarecolor{LightBlue}{84, 101, 202}
\definecolor{mydarkblue}{rgb}{0,0.08,0.45}
\hypersetup{ %
    pdftitle={},
    pdfkeywords={},
    pdfborder=0 0 0,
    pdfpagemode=UseNone,
    colorlinks=true,
    linkcolor=Navy,
    citecolor=DarkGreen,
    filecolor=Purple,
    urlcolor=Black,
}

\usepackage{fullpage}
\usepackage{authblk}

\usepackage{MnSymbol}

\usepackage{amsmath,amsfonts,bm}
\usepackage{multirow,array}
\usepackage{diagbox}

\def\vc{{\bm{c}}}

\def\vx{{\bm{x}}}
\def\vy{{\bm{y}}}
\def\vz{{\bm{z}}}

\def\vX{{\bm{X}}}
\def\vY{{\bm{Y}}}

\DeclareMathAlphabet{\mathsfit}{\encodingdefault}{\sfdefault}{m}{sl}
\SetMathAlphabet{\mathsfit}{bold}{\encodingdefault}{\sfdefault}{bx}{n}

\newcommand{\calA}{\ensuremath{\mathcal{A}}}

\newcommand{\calG}{\ensuremath{\mathcal{G}}}

\newcommand{\calX}{\ensuremath{\mathcal{X}}}
\newcommand{\calY}{\ensuremath{\mathcal{Y}}}

\newcommand{\R}{\mathbb{R}}

\renewcommand{\bar}[1]{\overline{#1}}

\newcommand{\norm}[1]{\left\| #1 \right\|}

\newcommand{\valpha}{\boldsymbol{\alpha} }
\newcommand{\vbeta}{\boldsymbol{\beta} }
\newcommand{\vgamma}{\boldsymbol{\gamma} }

\newcommand{\vzeta}{\boldsymbol{\zeta} }

\newcommand{\vtau}{\boldsymbol{\tau} }

\usepackage{amsthm}
\usepackage{mathtools}
\usepackage{amsmath}
\usepackage{bbm}
\usepackage{amsfonts}

\newcommand{\pureCircuit}{{\textsc{PureCircuit}}}

\newcommand{\linVI}{{\textsc{LinVI}}}

\newcommand{\Nor}{{\mathsf{NOR}} }
\newcommand{\Purify}{{\mathsf{PURIFY}}}
\newcommand{\Copy}{{\mathsf{COPY}}}
\newcommand{\Kcopy}{K_{\mathrm{copy}}}
\newcommand{\Fcopy}{\tilde{F}_{\mathrm{copy}}}
\newcommand{\Fblock}{\tilde{F}_{\mathrm{block}}}

\usepackage{nicefrac}
\usepackage{esvect}

\usepackage{tikz}
\usetikzlibrary{positioning, backgrounds, fit, arrows.meta, calc, decorations.pathreplacing, shapes.arrows, shadows, shapes, backgrounds, intersections, patterns.meta, quotes, angles}

\usepackage{pgfplots}
\pgfplotsset{compat=1.18}

\definecolor{uiGreenLight}{HTML}{E8F5E9}  
\definecolor{uiGreenMedium}{HTML}{C8E6C9} 
\definecolor{uiGreenDark}{HTML}{81C784}   
\definecolor{uiRedLight}{HTML}{FFEBEE}    
\definecolor{uiRedDark}{HTML}{E57373}     
\definecolor{uiText}{HTML}{263238}        

\usepackage{subfigure}

\usepackage[%
linewidth=2pt,
linecolor=gray,
middlelinecolor= black,
middlelinewidth=0.4pt,
roundcorner=1pt,
topline = false,
rightline = false,
bottomline = false,
rightmargin=0pt,
skipabove=0pt,
skipbelow=0pt,
leftmargin=0pt,
innerleftmargin=4pt,
innerrightmargin=0pt,
innertopmargin=0pt,
innerbottommargin=0pt,
]{mdframed}

\usepackage[capitalize,noabbrev,nameinlink]{cleveref}

\usepackage{enumitem}
\usepackage[dvipsnames]{xcolor}

\usepackage{bbm}

\usepackage{natbib}
\let\oldcitet\citet
\let\oldcitep\citep

\renewcommand{\citet}{\oldcitet*}
\renewcommand{\citep}{\oldcitep*}

\renewcommand{\vec}[1]{\bm{#1}}
\newcommand{\mat}[1]{\mathbf{#1}}

\usepackage{thm-restate}

\theoremstyle{plain}
\newtheorem{theorem}{Theorem}[section]
\newtheorem{lemma}[theorem]{Lemma}
\newtheorem{corollary}[theorem]{Corollary}

\theoremstyle{definition}
\newtheorem{definition}[theorem]{Definition}

\theoremstyle{remark}
\newtheorem{remark}[theorem]{Remark}

\title{The Computational Complexity of Team Zero-Sum Games}

\author[1]{Ioannis Anagnostides}
 \author[2]{Ioannis Panageas}
 \author[1,3]{Tuomas Sandholm}
 \author[2]{Jingming Yan}

 \affil[1]{Carnegie Mellon University}
 \affil[2]{University of California, Irvine}
 \affil[3]{\small Additional affiliations: Strategy Robot, Inc., Strategic Machine, Inc., Optimized Markets, Inc.}
 \affil[ ]{\texttt{\{ianagnos,sandholm\}}\texttt{@cs.cmu.edu}, \texttt{\{ipanagea,jingmy1\}}\texttt{@uci.edu}}

\begin{document}

\maketitle

\begin{abstract}
    A celebrated consequence of the minimax theorem is that two-player zero-sum games admit a tractable equilibrium characterization. In many central applications, however, each side comprises multiple independent agents who share a common objective but cannot perfectly coordinate their actions. Such settings can be modeled as \emph{team zero-sum games}, a natural generalization of both two-player zero-sum games and potential games---the two most well-studied classes of games in algorithmic game theory.

    In this paper, we settle the complexity of team zero-sum games by establishing that computing Nash equilibria is \PPAD-complete. As a result, despite the global adversarial structure, team zero-sum games are as hard as general-sum games. Our hardness result holds even when i) the precision is inverse polynomial, thereby ruling out a fully polynomial-time approximation scheme (unless $\P = \PPAD$); ii) each team consists of only two players; and iii) the underlying class of games is polymatrix. As a byproduct, we resolve the complexity of group-wise zero-sum polymatrix games, a class introduced and examined in the seminal work of Cai and Daskalakis (SODA '11), and more recently highlighted by Hollender, Maystre, and Nagarajan (ICLR '25). Moreover, we show that computing a first-order stationary point in min-max optimization is \PPAD-complete even for quadratic (multilinear) objectives.

    From a technical standpoint, we develop a series of team zero-sum game gadgets that allow us to simulate the breakthrough reduction of Bernasconi and Castiglioni (STOC '26). Moreover, to obtain hardness results for quadratic objectives, we make use of a general technique based on linear local approximation, which is of independent interest.
\end{abstract}


\clearpage

\section{Introduction}

Two-player zero-sum games constitute one of the central models of strategic interaction, going back to the foundational work of~\citet{vonNeumann28:Zur} and the inception of game theory. The minimax theorem endows two-player zero-sum games with a well-defined value and an unambiguous prescription of optimal play. Moreover, a minimax equilibrium---equivalently, a \emph{Nash equilibrium}~\citep{nash1951non}---can be computed in polynomial time via linear programming. This tractability has placed two-player zero-sum games at the forefront of operations research, optimization, and machine learning for many decades.

However, many strategic interactions involve multiple agents who share the same objective and compete against an opposing team of agents. Canonical examples include competition between firms~\citep{Coase37:Nature,Marschak55:Elements}, the weak selection model in evolutionary biology~\citep{Chastain14:Algorithms,Nowak04:Emergence}, adversarial machine learning with multiple classifiers or attackers~\citep{Durugkar17:Generative}, and multi-agent reinforcement learning~\citep{Tampuu15:Multiagent,Guan23:Zero,Kalogiannis23:Efficiently,Hu24:FM3Q,Zeng24:Computing,Kalogiannis24:Learning}. These settings can be naturally modeled as \emph{team zero-sum games}~\citep{Schulman19:Duality}. The key distinction from two-player zero-sum games is that each team comprises several \emph{independent} decision makers who may be unable to perfectly coordinate their actions. This reflects a practical constraint at the heart of group decision making: communication between the members of the team can be expensive or even infeasible~\citep{Coase37:Nature,Marschak55:Elements}; the popular card game bridge furnishes one such illustrative example. In fact, team zero-sum games generalize both two-player zero-sum games and \emph{potential games}~\citep{monderer1996potential}, the two most well-studied classes of games in algorithmic game theory.

In this paper, we examine the complexity of computing a Nash equilibrium---the standard solution concept in noncooperative games---in team zero-sum games. A celebrated series of works established that computing a Nash equilibrium even in two-player general-sum games is \PPAD-complete~\citep{Daskalakis09:The,Chen09:Settling}. However, those results do not apply to team zero-sum games since they are significantly more structured: there are only two diametrically opposing objectives, one for each team. That is, unlike general-sum games, team zero-sum games maintain a global adversarial structure. From a technical standpoint, this fundamentally precludes the use of existing techniques and reductions developed for general-sum games, which explains why the complexity of computing Nash equilibria in team zero-sum games has eluded prior research. A central question is to understand whether the tractability of two-player zero-sum games extends in the multiagent setting, or whether it is subject to the complexity barriers circumscribing general equilibrium computation.

A team zero-sum game can be viewed as a min-max optimization problem over a product of probability simplices. The key challenge is that once a team contains multiple players, the objective of the team becomes nonconvex/nonconcave in terms of their joint strategy. In this context, a recent breakthrough result by~\citet{Bernasconi26:Complexity} established that computing a \emph{first-order stationary point}---the counterpart of Nash equilibrium beyond finite games---in min-max optimization is \PPAD-hard even when the precision is inverse polynomial in the dimension. This establishes a stark separation between constrained optimization of a single function---wherein gradient descent efficiently converges to first-order stationary points---and min-max optimization. However, team zero-sum games are significantly more structured than general min-max optimization problems in terms of both the underlying objective and the constraint sets. This leaves open the tantalizing question of characterizing the complexity of team zero-sum games.

The complexity landscape of computing Nash equilibria is by now well-understood in the special case where one of the two teams contains a \emph{single} player. \citet{Anagnostides23:Algorithms} established membership in $\CLS$, thereby implying $\CLS$-completeness by virtue of the hardness result of~\citet{Babichenko21:Settling} concerning multi-player identical-interest games. More recently, \citet{Anagnostides25:Complexity} established $\CLS$-completeness even when a team of two players is competing against a single adversarial player. However, the more general setting beyond a single adversarial player remained open.

\subsection{Our results}

\begin{table}[t]
    \centering
    \caption{Complexity landscape for Nash equilibrium computation in team zero-sum games.}
    \label{tab:team-zero-sum-landscape}
    \small
    \renewcommand{\arraystretch}{1.35}
    \setlength{\tabcolsep}{5pt}
    \resizebox{\linewidth}{!}{%
    \begin{tabular}{@{}ccccc@{}}
        \toprule
        \textbf{Teams} 
        & \textbf{Setting} 
        & \textbf{Complexity}
        & \textbf{Precision}
        & \textbf{Reference} \\
        \midrule

        \rowcolor{green!6}
        $1$ vs.\ $1$
        & Two-player zero-sum games
        & \easybadge
        & Exact
        & Folklore \\

        \rowcolor{yellow!6}
        $n$ vs.\ $1$, $n \ge 2$
        & Adversarial team games
        & \clsbadge
        & $1/\mathsf{exp}$
        & \oldcitet{Anagnostides23:Algorithms,Anagnostides25:Complexity} \\

        \rowcolor{red!6}
        $n$ vs.\ $m$, $n,m \ge 2$
        & Team zero-sum games
        & \ppadbadge
        & $1/\poly$
        & \textbf{This paper} \\

        \bottomrule
    \end{tabular}%
    }
\end{table}

Our main contribution is to close this gap, establishing that team zero-sum games are as hard as general-sum games~\citep{Daskalakis09:The,Chen09:Settling}.

\begin{theorem}
    \label{theorem:informal-team}
    Computing an $\epsilon$-approximate Nash equilibrium in $2$ vs. $2$ team zero-sum games is \PPAD-complete even when $\epsilon$ is inversely polynomial in the size of the input.
\end{theorem}

As a result, we show that the complexity of team zero-sum games collapses to that of $2$ vs. $2$ team zero-sum games. Together with prior results, we thus paint a complete picture of the complexity landscape in team zero-sum games. \Cref{theorem:informal-team} also rules out the existence of a fully polynomial-time approximation scheme (unless $\P = \PPAD$).

Taking a step further, we establish \PPAD-completeness in \emph{polymatrix} (\emph{e.g.}, \citealp{Cai11:Minmax,Cai16:Zero}) team zero-sum games.

\begin{theorem}
    \label{theorem:informal-polymatrix}
    Computing an $\epsilon$-approximate Nash equilibrium in polymatrix team zero-sum games is \PPAD-complete even when $\epsilon$ is inversely polynomial in the size of the input.
\end{theorem}

As a byproduct, our results resolve the complexity of \emph{group-wise} zero-sum polymatrix games, a class introduced and examined in the seminal work of~\citet{Cai11:Minmax}. In their \PPAD-hardness result, the players can be partitioned into \emph{three} groups such that the games played inside each group are coordination (\emph{i.e.}, identical-interest) whereas the games played across groups are zero-sum. We show that \PPAD-hardness persists even when there are two groups of players. Closing this gap was the main open problem recently highlighted by~\citet{Hollender25:Complexity}.

\paragraph{Implications for min-max optimization} Turning to the problem of min-max optimization, \Cref{theorem:informal-polymatrix} implies that the \PPAD-hardness of~\citet{Bernasconi26:Complexity} holds even when the underlying objective function is a degree-$2$ (multilinear) polynomial. 

\begin{corollary}
    Computing an $\epsilon$-first-order stationary point in min-max optimization where the objective is a quadratic polynomial is \PPAD-complete even when $\epsilon$ is inversely polynomial in the dimension.
\end{corollary}

\subsection{Technical approach} 

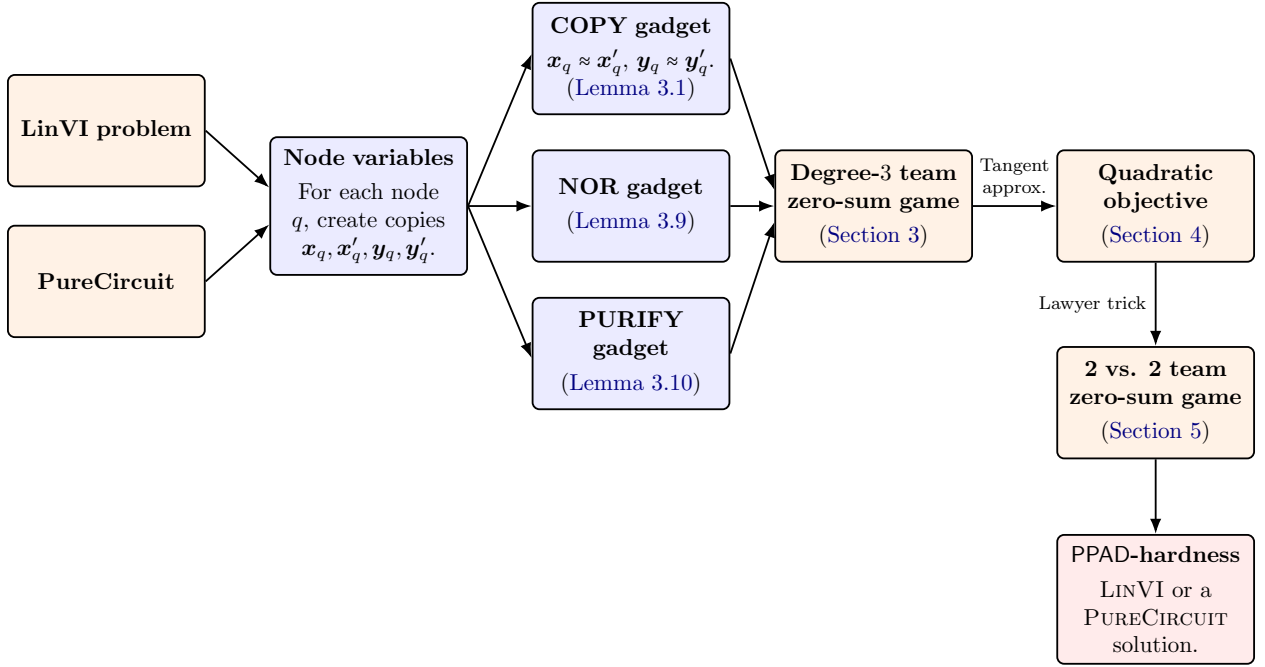
\begin{figure*}[!t]
\raggedright
\makebox[\textwidth][l]{%
\resizebox{\textwidth}{!}{%
\begin{tikzpicture}[
    >=Latex,
    every node/.style={font=\small},
    box/.style={
        draw,
        rounded corners=3pt,
        thick,
        align=center,
        inner sep=5pt,
        text width=2.65cm,
        minimum height=1.7cm
    },
    src/.style={box, fill=orange!10},
    gad/.style={box, fill=blue!8},
    proc/.style={box, fill=orange!10},
    final/.style={box, fill=red!8},
    arr/.style={->, thick}
]

\node[src] (linvi) at (0,1.15) {
    \textbf{\linVI{} problem} \par\smallskip
};

\node[src] (pc) at (0,-1.15) {
    \textbf{\pureCircuit{} }\par\smallskip
};

\node[gad] (vars) at (4.0,0) {
    \textbf{Node variables}\par\smallskip
    For each node $q$, create copies
    $\vx_q,\vx'_q,\vy_q,\vy'_q$.
};

\node[gad] (copy) at (8.0,2.25) {
    \textbf{COPY gadget}\par\smallskip
    $\vx_q \approx \vx'_q$,
    $\vy_q \approx \vy'_q$.\par
    (\Cref{lem:copy_consist})
};

\node[gad] (nor) at (8.0,0) {
    \textbf{NOR gadget}\par\smallskip
    (\Cref{lem:nor_correct})
};

\node[gad] (purify) at (8.0,-2.25) {
    \textbf{PURIFY gadget}\par\smallskip
    (\Cref{lem:puri_correct})
};

\node[proc] (degthree) at (11.7,0) {
    \textbf{Degree-$3$ team zero-sum game}\par\smallskip
    (\Cref{sec:degree3})
};

\node[proc] (degtwo) at (16,0) {
    \textbf{Quadratic objective}\par\smallskip
    (\Cref{sec:quadratic})
};

\node[proc] (layer) at (16.0,-3.00) {
    \textbf{2 vs. 2 team zero-sum game}\par\smallskip
    (\Cref{sec:lawyer})
};

\node[final] (hardness) at (16.0,-6.0) {
    \textbf{$\PPAD$-hardness}\par\smallskip
    \linVI{}
    or a \pureCircuit{} solution.
};


\draw[arr] (linvi.east) -- ([yshift=8pt]vars.west);
\draw[arr] (pc.east)    -- ([yshift=-8pt]vars.west);

\draw[arr] (vars.east) -- ([yshift=2pt]copy.west);
\draw[arr] (vars.east) -- (nor.west);
\draw[arr] (vars.east) -- ([yshift=-2pt]purify.west);

\draw[arr] (copy.east)   -- ([yshift=7pt]degthree.west);
\draw[arr] (nor.east)    -- (degthree.west);
\draw[arr] (purify.east) -- ([yshift=-7pt]degthree.west);

\draw[arr] (degthree.east) -- node[above, midway, font=\scriptsize, align=center] {Tangent \\ approx.} (degtwo.west);

\draw[arr] (degtwo.south) -- node[left, midway, font=\scriptsize] {Lawyer trick} (layer.north);
\draw[arr] (layer.south) -- (hardness.north);

\end{tikzpicture}%
}%
}
\caption{
Overview of the reduction.
}
\label{fig:reduction_overview}
\end{figure*}

We now give an overview of the reduction (\Cref{fig:reduction_overview}). Following~\citet{Bernasconi26:Complexity}, our starting point is an instance of $\pureCircuit$ together with an instance of \linVI; solving \emph{either} of these problems is \PPAD-hard. Our reduction constructs a team zero-sum game whose Nash equilibria encode either a solution to the embedded $\linVI$ instance or a valid assignment to $\pureCircuit$. In a nutshell, we proceed in three stages. First, we construct a degree-$3$ multilinear min-max objective using a series of team zero-sum game gadgets which guarantees the consistency of the reduction. Second, we show that the degree-$3$ objective can be simulated by a quadratic multilinear objective through local approximation. Finally, we convert the resulting quadratic min-max optimization problem over a product of simplices into a $2$ vs.\ $2$ team zero-sum game.

\paragraph{Encoding circuit nodes} For every circuit node $q\in V$, we introduce multiple copies $\vx_q,\vx'_q,\vy_q,\vy'_q\in[0,1]^{mn}$, where $m \gg n$. The purpose of the primed variables is that they allow us to introduce the \emph{multilinear} term $V_q = \sum_{r\in[m],i\in[n]} (x_{q,r,i}-y_{q,r,i})(x'_{q,r,i}-y'_{q,r,i})$, which behaves like $\|\vx_q-\vy_q\|_2^2$---per the original reduction of~\citet{Bernasconi26:Complexity}---when $\vx_q \approx \vx_q'$ and $\vy_q \approx \vy_q'$. We also consider the linking term $H_q = \sum_{r\in[m]} \langle \mat{D} \vx'_{q,r}+\vc,\vx_{q,r}-\vy_{q,r}\rangle$, which is where the embedded $\linVI$ instance enters the construction.

\paragraph{The three gadgets.} The construction uses three main gadgets. The first is the $\Copy$ gadget, which forces $\vx_q \approx \vx'_q$ and $\vy_q \approx \vy'_q$ at every approximate equilibrium. As discussed, this guarantees that $V_q$ is a good proxy for $\|\vx_q-\vy_q\|_2^2$. The second gadget implements the NOR gate. If $u,v$ are the inputs and $w$ is the output, we introduce a simplex variable $\valpha^{u,v,w}\in\Delta_3$. The first two coordinates of $\valpha^{u,v,w}$ test whether $V_u$ and $V_v$ exceed the designated threshold, while the third coordinate activates the output node through the linking term $H_w$. Thus, when both inputs are inactive, the third action is optimal and the output value $s_w$ becomes close to $1$; in the contrary case, when either input is active, the third action is dominated and $s_w$ becomes close to $0$. To ensure that $\valpha^{u,v,w}$ has the desired behavior, we scale the utilities of the first two coordinates by a large coefficient $K$. This is a standard technique in gadget constructions \citep{Babichenko21:Settling, Fearnley25:Complexity}, which effectively forces $\valpha^{u,v,w}$ to behave as a function of $V_u$ and $V_v$. The consistency of this gadget is formalized in \Cref{lem:nor_correct}. The third gadget implements the PURIFY gate. If $u$ is the input and $v,w$ are the outputs, we use two binary simplex variables to determine the values of the outputs. The thresholds $3n+1/4$ and $3n+3/4$ ensure that if $u$ is inactive then both outputs are inactive, and if $u$ is active then both outputs are active. In the intermediate region at least one of the two outputs is forced to a definite Boolean value. This gives the desired behavior in \Cref{lem:puri_correct}.

\paragraph{Reducing to a quadratic polynomial} The objective we have constructed so far contains degree-3 monomials involving three independent maximizers. This is an obstacle in establishing the desired hardness for 2 vs.\ 2 team zero-sum games, since one team would need to contain at least $3$ players. To address this challenge, our next step is to eliminate the degree-$3$ terms. To do so, we use an algebraic identity to write every cubic monomial $\tilde{x}\tilde{y}\tilde{z}$ as a linear combination of univariate cubic functions of averages of $\tilde{x},\tilde{y},\tilde{z}$. We then approximate each univariate cubic through a discretized family of tangent functions (\Cref{fig:tangent}). In particular, we introduce auxiliary variables whose role is to select the correct tangent. It can be shown that, with sufficiently fine discretization and sufficiently small precision, any approximate stationary point of the quadratic objective yields an approximate stationary point of the original degree-$3$ objective.\footnote{A reader may wonder whether this technique can be applied to other problems, most notably computing KKT points in constrained optimization~\citep{Fearnley25:Complexity}. The challenge there is that the precision needs to be exponentially small to obtain $\CLS$-hardness, so our approach cannot be directly applied.}

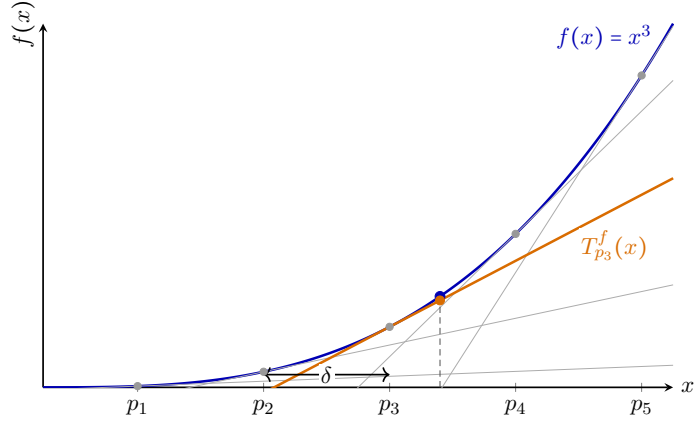
\begin{figure}[t]
    \centering
    \scalebox{0.8}{\begin{tikzpicture}
  \pgfmathsetmacro{\pa}{0.15}
  \pgfmathsetmacro{\pb}{0.35}
  \pgfmathsetmacro{\pc}{0.55}
  \pgfmathsetmacro{\pd}{0.75}
  \pgfmathsetmacro{\pe}{0.95}
  \pgfmathsetmacro{\tq}{0.63}
  \pgfmathsetmacro{\pt}{\pc}
  \pgfmathsetmacro{\ft}{\tq^3}
  \pgfmathsetmacro{\ttan}{\pt^3 + 3*(\tq-\pt)*\pt^2}

  \begin{axis}[
    width=12cm,
    height=7.6cm,
    xmin=0, xmax=1,
    ymin=0, ymax=1,
    enlargelimits=false,
    axis lines=left,
    axis line style={thick},
    xlabel={$x$},
    ylabel={$f(x)$},
    xlabel style={at={(axis description cs:1,0)},anchor=west},
    ylabel style={at={(axis description cs:0,1)},anchor=south},
    xtick={\pa,\pb,\pc,\pd,\pe},
    xticklabels={$p_1$,$p_2$,$p_3$,$p_4$,$p_5$},
    ytick=\empty,
    clip=true
  ]
    \addplot[domain=0:1,samples=300,very thick,blue!70!black] {x^3};
    \node[blue!70!black,anchor=south east,fill=white,inner sep=1pt]
      at (axis cs:0.97,0.92) {$f(x)=x^3$};

    \foreach \p in {\pa,\pb,\pc,\pd,\pe} {
      \addplot[domain=0:1,samples=2,gray!65,thin] {\p^3 + 3*(x-\p)*\p^2};
      \addplot[only marks,mark=*,mark size=1.7pt,gray!80] coordinates {(\p,{\p^3})};
    }

    \addplot[domain=0:1,samples=2,very thick,orange!85!black] {\pt^3 + 3*(x-\pt)*\pt^2};
    \node[orange!85!black,anchor=west,fill=white,inner sep=1pt]
      at (axis cs:0.85,0.39) {$T_{p_3}^{f}(x)$};

    \addplot[densely dashed,black!70] coordinates {(\tq,0) (\tq,\ft)};
    \addplot[only marks,mark=*,mark size=2.3pt,blue!70!black] coordinates {(\tq,\ft)};
    \addplot[only marks,mark=*,mark size=2.1pt,orange!85!black] coordinates {(\tq,\ttan)};


    \draw[<->,thick] (axis cs:\pb,0.035) -- (axis cs:\pc,0.035)
      node[midway,fill=white,inner sep=1pt] {$\delta$};

  \end{axis}
\end{tikzpicture}}
    \caption{A schematic illustration of the gadget that enables reducing to a quadratic multilinear objective. We construct a grid $\calG$ with mesh $\delta \ll 1$ so that if $T_p^f$ is the tangent at some point $p$, $x^3$ is well approximated by the piecewise linear function $\max_{p \in \calG} T_p^f$. In turn, we introduce an auxiliary variable $\vec{\zeta}$, and consider the function $\sum_{p \in \calG} \vec{\zeta}_p T_{p}^f$. The upshot is that stationarity forces $\vec{\zeta}$ to select the correct tangent, as formalized in~\Cref{sec:quadratic}.}
    \label{fig:tangent}
\end{figure}

\paragraph{Reducing to $2$ vs.\ $2$ team zero-sum games} Finally, we convert the problem over products of simplices into a $2$ vs.\ $2$ team zero-sum game. We introduce two minimizers $\tilde{\vX},\tilde{\vX}'$ and two maximizers $\tilde{\vY},\tilde{\vY}'$. The first ingredient is a $\Copy$-style gadget that ensures the two players on the same side encode consistent copies of the original variables. The second gadget follows the ``lawyer trick'' \emph{\`a la}~\citet{Daskalakis09:The}, which ensures that blocks corresponding to the original simplices get allocated approximately equal probability mass. The key claim is that an approximate Nash equilibrium of the resulting 2 vs.\ 2 team zero-sum game gives an approximate first-order stationary point of the original problem.

\subsection{Further related work}

The study of adversarial team games goes back to the seminal work of~\citet{Stengel97:Team}, who introduced the notion of a \emph{team maxmin equilibrium (TME)}, which can be viewed as the \emph{best} Nash equilibrium of the team. However, computing a TME is \NP-hard~\citep{Borgs10:The}. This has shifted the research focus to relaxations of TME that make allowance for \emph{ex ante} correlation between the members of the team, a direction that has proven particularly fruitful in the context of extensive-form games~\citep{Basilico17:Team,Farina18:Ex,farina2018ex,Farina21:Connecting,Zhang20:Computing}. More recently, \citet{Carminati23:Hidden} considered \emph{hidden role games}, which is a considerably more challenging class of problems. Our focus in this paper is on the standard Nash equilibrium.

Beyond settings with a single adversary, \citet{Schulman19:Duality} analyzed the \emph{duality gap} in team zero-sum games as a function of the number of players. \citet{Kalogiannis23:Teamwork} established that certain natural learning algorithms fail to converge even in $2$ vs. $2$ team zero-sum games; our main hardness result provides an explanation for this phenomenon. \citet{Hollender25:Complexity} established $\CLS$-hardness for computing Nash equilibria in team zero-sum games with multiple \emph{independent} adversaries. $\CLS$, which stands for continuous local search, is a complexity class introduced by~\citet{Daskalakis11:Continuous}. \citet{Fearnley23:Complexity} showed that, surprisingly, $\CLS = \PPAD \cap \PLS$. $\PPAD$ was introduced by~\citet{Papadimitriou94:Complexity}, and was famously shown to characterize the complexity of computing Nash equilibria in general-sum games~\citep{Daskalakis09:The,Chen09:Settling}. One question that remains open in team zero-sum games concerns the regime where the precision $\epsilon$ is an absolute constant. \citet{Rubinstein16:Settling} established a quasipolynomial lower bound in general-sum games under a natural complexity assumption; extending that result in team zero-sum games with a constant number of players is an important open question.

Turning to the more general setting of min-max optimization, \citet{DSZ21} established \PPAD-hardness for computing first-order stationary points under \emph{coupled constraints}. Subsequent work simplified their reduction~\citep{Anagnostides25:Complexity,bernasconi2024role}. A recent breakthrough result by~\citet{Bernasconi26:Complexity} showed \PPAD-hardness even with uncoupled constraints, which is more in line with the applications of min-max optimization. Even more recently, \citet{Bernasconi26:Minmax} established exponential query lower bounds in the black box model for computing first-order stationary points in min-max optimization. Motivated by a host of applications in machine learning and reinforcement learning, there has been a tremendous amount of interest in solving min-max optimization problems; \emph{e.g.}, we refer to~\citet{Ostrovskii21:Efficient,Jordan23:First,Nouiehed19:Solving,Wibisono22:Alternating,MokhtariOP20,Choudhury23:Single,Stochastic20:Luo,Xu23:Unified,Gorbunov22:Last,Cai22:Finite,daskalakis2020independent}, and references therein.

\paragraph{Concurrent and independent work} An independent and concurrent work by~\citet{Bernasconi26:ComplexityQuadratic} also establishes \PPAD-hardness for min-max optimization with quadratic objectives. We will elaborate on the technical differences between the two approaches in a future version of this paper.
\section{Preliminaries}

\paragraph{Team zero-sum games} Our primary focus in this paper is on \emph{team zero-sum games}. We work in the usual normal-form representation. We assume that the game contains $[n+m] \defeq \{1, \dots, n+m\}$ players. Each player $i \in [n+m]$ selects as strategy a probability distribution over a finite set of available actions $\calA_i$. Each player $i$ has a utility function $u_i : \calA_1 \times \dots \times \calA_{n+m} \to \R$. In an $n$ vs. $m$ team zero-sum game, there exists a utility function $u: \calA_1 \times \dots \times \calA_{n+m} \to \R$ such that for any joint action profile $\vec{a} = (a_1, \dots, a_{n+m})$, we have $u_i(\vec{a}) = u(\vec{a})$ for any $i \in [n]$ and $u_i(\vec{a}) = - u(\vec{a})$ for any $i \in [n+1, n+m] \defeq \{n+1, \dots, n+m\}$. When we are dealing with a game with a constant number of players, the utility function $u$ will be given explicitly as part of the input. More generally, $u$ will be described succinctly as a (multilinear) polynomial of constant degree. In particular, when $u$ is a quadratic (multilinear) polynomial, the game is referred to as \emph{polymatrix}~\citep{Cai11:Minmax,Cai16:Zero}.

\begin{definition}[\citealp{Nash50:Non}]
    \label{def:Nash-eq}
    A joint strategy $(\vx_1, \dots, \vx_{n+m}) \in \Delta(\calA_1) \times \dots \times \Delta(\calA_{n+m})$ is an \emph{$\epsilon$-approximate Nash equilibrium} if for any player $i \in [n+m]$ and strategy $\vx'_i \in \Delta(\calA_i)$,
    \[
        u_i(\vx_{i}, \vx_{-i}) \geq u_i(\vx'_i, \vx_{-i}) - \epsilon,
    \]
    where $\vx_{-i} \defeq (\vx_1, \dots, \vx_{i-1}, \vx_{i+1}, \dots, \vx_{n+m})$.
\end{definition}

\paragraph{Min-max optimization} A more general setting concerns the computation of (first-order) stationary points in min-max optimization. Here, we are given a smooth objective function $f : \calX \times \calY \to \R$, where $\calX$ and $\calY$ are convex and compact strategy sets. It is assumed that Player $\calX$ is striving to minimize $f$ whereas Player $\calY$ is striving to maximize it. The counterpart to~\Cref{def:Nash-eq} in min-max optimization is summarized below.

\begin{definition}[First-order stationary points]
    \label{def:fosp}
    A point $(\vx, \vy) \in \calX \times \calY$ is an \emph{$\epsilon$-first-order stationary point} if for any $(\vx', \vy') \in \calX \times \calY$,
    \[
        \langle \vx' - \vx, \nabla_{\vx} f(\vx, \vy) \rangle \geq - \epsilon, \quad \text{and} \quad \langle \vy - \vy', \nabla_{\vy} f(\vx, \vy) \rangle \geq - \epsilon.
    \]
\end{definition}

Moving on, following~\citet{Bernasconi26:Complexity}, our reduction uses two \PPAD-hard instances, which we introduce below.

\paragraph{Pure circuit}
A $\pureCircuit$ instance has vertex set $V$ and two types of gates, namely $\Nor$ and $\Purify$. Every node in the circuit is the output of exactly one gate. A solution to $\pureCircuit$ consists of assigning to every node a value from the set $\{0,1,\bot\}$ so that the following requirements are met. Let $G_{\Nor}$ and $G_{\Purify}$ be the sets of all $\Nor$ and $\Purify$ gates in a $\pureCircuit$. For a $\Nor$ gate $(u,v,w) \in G_{\Nor}$ with inputs \{$ u,v \}$ and output $w$,
\[
\lambda(u)=\lambda(v)=0\Rightarrow \lambda(w)=1,
\qquad
\lambda(u)=1\text{ or }\lambda(v)=1\Rightarrow \lambda(w)=0.
\]
For a $\Purify$ gate $(u,v,w) \in G_{\Purify}$ with input $u$ and outputs $ \{ v,w \}$,
\[
\lambda(u)\in\{0,1\}\Rightarrow \lambda(v)=\lambda(w)=\lambda(u),
\]
and when $\lambda(u)=\bot$, at least one of $\lambda(v), \lambda(w)$ is in $\{0,1\}$.

\begin{theorem}[\citealp{Deligkas24:Pure}]
    \label{thm:purecircuit}
    $\pureCircuit$ is \PPAD-complete.
\end{theorem}

\PPAD-hardness holds even when the out-degree $\kappa$ is constant~\citep{Deligkas24:Pure}.

\paragraph{Linear variational inequalities}

A $\linVI$ instance over the hypercube $[0, 1]^n$ is specified by a matrix $\mat{D} \in [-1,1]^{n\times n}$ such that $\norm{\mat{D}}_{1} \defeq \max_{j \in [n]} \sum_{i \in [n]} |\mat{D}_{i,j}| \leq 1 $ and $\norm{\mat{D}}_{\infty} \defeq \max_{i \in [n]} \sum_{j \in [n]} |\mat{D}_{i,j}| \leq 1 $, a vector $\vec{c} \in [-1,1]^n$, and a parameter $\rho\in(0,1]$. The $\linVI$ problem asks for a point $\vec{z} \in [0,1]^n$ such that
\[
(\mat{D} \vec{z} + \vec{c} )_i(z_i'-z_i)\ge -\rho
\quad\forall i\in[n],\ \forall z_i'\in[0,1].
\]

\begin{theorem}[\citealp{Rubinstein15:Inapproximability,bernasconi2024role}]
    \label{theorem:LinVI}
    $\linVI$ is \PPAD-hard even when $\rho > 0$ is an absolute constant.
\end{theorem}

\paragraph{Notation}
For the rest of the paper, we use $n$ to denote the dimension of a $\linVI$  instance, i.e., $\mat{D} \in [-1, 1]^{n \times n}.$ For a vector $\vx \in \R^n$, we denote by $\norm{\vx}_2$, $\norm{\vx}_1$, and $\norm{\vx}_\infty$ the Euclidean norm, the $\ell_1$ norm, and $\ell_\infty$ norm, respectively. We sometimes use the $O(\cdot), \Theta(\cdot), \Omega(\cdot)$ notation to suppress absolute constants.
\section{Min-max optimization with degree-$3$ objective}
\label{sec:degree3}

In this section, we establish \PPAD-completeness for team zero-sum games in which the utility function is a degree-$3$ multilinear function.

\subsection{Implementing the gates}

We begin by constructing functions in $\pureCircuit$ using team zero-sum game gadgets. For each node $q\in V$ in the $\pureCircuit$, copy $r \in [m]$, and coordinate $i \in [n]$, with $m = n^{16}.$ We define minimizing variables $\vx_{q}, \vx'_{q} \in [0, 1]^{mn},$
and maximizing variables
$\vy_{q}, \vy'_{q} \in [0, 1]^{mn}.$

For any fixed node $q$, we have $\vx_q, \vx'_q, \vy_q, \vy'_q \in [0, 1]^{mn}.$
We introduce the $\Copy_q$ gadget for every node $q \in V$ that enforces $\vx_q \simeq \vx'_q$ and $\vy_q \simeq \vy'_q.$ Let $\vz \in \Delta_{2mn+1}$ be a minimizer and $\vz' \in \Delta_{2mn+1}$ be a maximizer, the $\Copy$ gadget for $q$ is defined as
\begin{align*}
     \Copy_q \defeq &\sum_{r \in [m], i \in [n]}\left( z'_{q, r, i, +}(x_{q, r, i} - x'_{q,r,i}) + z'_{q, r, i, -} (x'_{q,r,i} - x_{q,r,i})\right) + z'_{0} \cdot 0 \\
     & +  \sum_{r \in [m], i \in [n]}\left( z_{q,r,i,+}(y_{q,r,i} - y'_{q,r,i}) + z_{q,r,i,-} (y'_{q,r,i} - y_{q,r,i})\right) + z_{0} \cdot 0.
\end{align*}
For every node $q \in V$, we define
\begin{align}
    \label{align:Vq}
    V_q \defeq \sum_{r \in [m], i \in [n]} (x_{q,r,i} - y_{q,r,i})(x'_{q,r,i} - y'_{q,r,i}).
\end{align}
When $\vx'_q \simeq \vx_q$ and $\vy'_q \simeq \vy_q$, the value $V_q$ approximates $\norm{\vx_q - \vy_q}^2.$ Furthermore, for every node $q$, we define the linking function $H_q$ as
\begin{align*}
    H_q \defeq \sum_{r \in [m]} \langle \mat{D}\vx'_{q,r} + \vc, \vx_{q,r} - \vy_{q,r}\rangle.
\end{align*}

For any $\Nor$ gate with input nodes $u,v$ and output node $w$, we introduce a maximizing variable $\valpha^{u, v, w} \in \Delta_3$ and define the $\Nor$ gate gadget to be
\begin{align*}
     \Nor_{u, v, w} \defeq & K\left(\alpha^{u ,v, w}_{1} (V_u -3n - \frac{1}{2}) + \alpha^{u,v,w}_2 (V_v -3n - \frac{1}{2})\right) +\alpha_3^{u, v, w} H_w,
\end{align*}
where $K = 10^4\kappa^2 n^9 \log(m)$ is a scaling parameter independent of $\epsilon$, and $\kappa$ denotes the maximum out-degree of $\pureCircuit$. Intuitively, the value of $\alpha^{u, v, w}_3$ serves as the output value of node $w$ from the $\Nor$ gate.

For any $\Purify$ gate with input node $u$ and output nodes $v, w$, we introduce two maximizing variables $\vbeta^{u, v, w}, \vgamma^{u, v, w} \in \Delta_2$ and define the gadget to be
\begin{align*}
         \Purify_{u,v,w} \defeq & K \beta^{u, v, w}_1 \left( V_u - 3n -\frac{1}{4}\right) + \beta_1^{u, v, w} H_v  + K \gamma^{u, v, w}_1 \left( V_u - 3n -\frac{3}{4}\right) + \gamma_1^{u, v, w} H_w.
    \end{align*}
The intended values for output nodes $v, w$ are $\beta_1^{u ,v, w}$ and $\gamma_1^{u, v, w}$ respectively. Since for each node $q \in V$, it is the input of at most $\kappa$ gates and the output of exactly one gate, the value for node $q$ is uniquely determined, and we denote it as $s_q$. Specifically, we define
\begin{align*}
    s_q \defeq \sum_{(u,v,q) \in G_{\Nor}} \alpha_3^{u,v,q} + \sum_{(u,q,w) \in G_{\Purify}} \beta_1^{u,q,w} + \sum_{(u,v,q) \in G_{\Purify}} \gamma_1^{u,v,q}.
\end{align*}
Exactly one term on the right side of the equation above is non-zero and $0 \leq s_q \leq 1$. Finally, for every node $q \in V$, for every copy $r \in [m],$ let $\delta = \kappa^{-2}n^{-6}$, we define
\begin{align}\label{eq:M_def}
    M_{r} \defeq \delta\left(-\frac{m}{2} + r\right),
\end{align}
and the regularization
\begin{align*}
    \Phi = \sum_{q \in V} \sum_{r \in [m]} M_{r} \sum_{i \in [n]} (x_{q,r,i} - y_{q,r,i})(x'_{q,r,i} - y'_{q,r,i}).
\end{align*}
The final objective in our team zero-sum game is 
\begin{align}\label{eq:final_game}
    F \defeq \sum_{(u, v, w) \in G_{\Nor}} \Nor_{u, v, w} + \sum_{(u, v, w) \in G_{\Purify}} \Purify_{u, v, w} + \Phi + \Kcopy\sum_{q \in V} \Copy_q,
\end{align}
where we set the parameter $\Kcopy = 4mn\kappa (K + n + m\delta)/\epsilon$. Furthermore, for each node $q \in V$, we define the following function
\begin{align*}
    F_{q, \text{in}} = \sum_{(q, v, w) \in G_{\Nor}} \Nor_{q, v, w} + \sum_{(u, q, w) \in G_{\Nor}} \Nor_{u, q, w} + \sum_{(q, v, w) \in G_{\Purify}} \Purify_{q, v, w}
\end{align*}
to be the total sum of all gate gadgets where $q$ is the input of the gate. Since each node $q$ is the input to at most $\kappa$ gates in \pureCircuit, $F_{q, in}$ has at most $\kappa$ non-zero terms. Similarly, we define 
\begin{align*}
    F_{q, \text{out}} = \sum_{(u, v, q) \in G_{\Nor}} \Nor_{u, v, q} + \sum_{(u, q, w) \in G_{\Purify}} \Purify_{u, q, w} + \sum_{(u, v, q) \in G_{\Purify}} \Purify_{u, v, q}
\end{align*}
to be the sum of gate gadgets where $q$ is the output of the gate. Since each node $q$ is exactly the output of one gate, only one term in $F_{q, out}$ is non-zero.

\subsection{Copy consistency}

Next, we prove a lemma concerning the $\Copy$ gadget that will help us to establish the correctness of each gadget we implement.

\begin{lemma} \label{lem:copy_consist}
    For every $\epsilon$-first-order stationary point of the objective defined in \eqref{eq:final_game} and for any node $q \in V$, it holds that
    \begin{align*}
        \norm{\vx^*_q - \vx'^*_q}_\infty \leq \frac{4\epsilon}{\Kcopy}, \quad \text{and} \quad \norm{\vy^*_q - \vy'^*_q}_\infty \leq \frac{4\epsilon}{\Kcopy}.
    \end{align*}
\end{lemma}

\begin{proof}
    Observe that for the maximizer $\vz_q'$, by putting all probability mass on the coordinate $z'_{q,r, i, +}$ or $z'_{q,r,i, -}$ where $|x^*_{q, r, i} - x'^*_{q, r, i}| = \norm{\vx^*_{q} - \vx'^*_{q}}_\infty$, we have 
    \begin{align*}
        \sum_{r \in [m], i \in [n]}\left( z'_{q, r, i, +}(x^*_{q, r, i} - x'^*_{q,r,i}) + z'_{q, r, i, -} (x'^*_{q,r,i} - x^*_{q,r,i})\right) + z'_{0} \cdot 0 = \norm{\vx^*_{q} - \vx'^*_{q}}_\infty.
    \end{align*}
    Thus, from the $\epsilon$-first-order stationary condition of $\vz'^*_q$, it holds that
    \begin{align}\label{eq:copy_value_after}
        \Kcopy \left(\sum_{r \in [m], i \in [n]}\left( z'^*_{q, r, i, +}(x^*_{q, r, i} - x'^*_{q,r,i}) + z'^*_{q, r, i, -} (x'^*_{q,r,i} - x^*_{q,r,i})\right) + z'^*_{0} \cdot 0 \right)\geq  \Kcopy\norm{\vx^*_{q} - \vx'^*_{q}}_\infty - \epsilon.
    \end{align}
    Now consider the deviation of $\vx_q^*$ to $\vx_q'^*.$ Since $\vx_q^*$ is part of an $\epsilon$-first-order stationary point, it holds that 
    \begin{align} \label{eq:total_value_after}
        \langle \vx_q'^* - \vx_q^*, \nabla_{\vx^*_q} F_{q, \text{in}} + \nabla_{\vx^*_q} F_{q, \text{out}} + \nabla_{\vx^*_q} \Phi + \Kcopy\nabla_{\vx^*_q} \Copy_q\rangle \geq -\epsilon.
    \end{align}
    We aim to bound each term individually. First,
    \begin{align*}
        \norm{\nabla_{\vx^*_{q}} F_{q,in}}_\infty \leq  & \norm{\sum_{(q, v, w) \in G_{\Nor}} K \alpha^{q, v, w, *}_1 \norm{\vx'^*_q - \vy'^*_q}_\infty + \sum_{(u, q, w) \in G_{\Nor}} K \alpha^{u, q, w, *}_2 \norm{\vx'^*_q - \vy'^*_q}_\infty
        }_\infty \\
        & + \norm{\sum_{(q, v, w) \in G_{\Purify}} K (\beta^{q, v, w, *}_1 + \gamma^{q, v, w,*}_1) \norm{\vx'^*_q - \vy'^*_q}_\infty}_\infty \\
        \leq & 2\kappa K,
    \end{align*}
    where in the last step we use the fact that $q$ is the input to at most $\kappa$ gates, $|\beta^{q, v, w, *}_1 + \gamma^{q, v, w, *}_1| \leq 2$, and $\norm{\vx'^*_q - \vy'^*_q}_\infty \leq 1.$ Similarly, for $F_{q, out},$ we have
    \begin{align*}
        \norm{\nabla_{\vx^*_{q,r}} F_{q,out}}_\infty \leq & \norm{\sum_{u, v, q \in G_{\Nor}} \alpha^{u,v,q,*}_3\left(\mat{D}\vx'^*_{q,r} + \vc\right)}_\infty + \norm{\sum_{u, q, w \in G_{\Purify}} \beta^{u,q,w, *}_1\left(\mat{D}\vx'^*_{q,r} + \vc\right)}_\infty \\
        & + \norm{\sum_{u, v, q \in G_{\Purify}} \gamma^{u,v,q, *}_1\left(\mat{D}\vx'^*_{q,r} + \vc\right)}_\infty \\
        \leq & 2n,
    \end{align*}
    where in the last step we use the fact that $\norm{\mat{D}\vx'^*_{q,r} + \vc}_\infty \leq n + 1 \leq 2n$ and $q$ is the output for exactly one gate. For the regularization term $\Phi,$ it holds that
    \begin{align*}
    \norm{\nabla_{\vx^*_q}\Phi}_\infty \leq m\delta \norm{\vx'^*_q - \vy'^*_q}_\infty \leq m\delta.
    \end{align*}
    From \eqref{eq:copy_value_after} and the fact that after the deviation of $\vx^*_q$ to $\vx'^*_q$, the value of the term involving $\vx_q$ and $\vx'_q$ in the $\Copy$ gadget becomes $0.$ Therefore, we conclude that
    \begin{align*}
        \langle\vx'^*_q - \vx^*_q, \Kcopy \nabla_{\vx^*_q} \Copy_q\rangle \leq -\Kcopy \norm{\vx^*_q - \vx'^*_q}_\infty + \epsilon.
    \end{align*}
    Thus
    \begin{align*}
        & \langle \vx_q'^* - \vx_q^*, \nabla_{\vx^*_q} F_{q, \text{in}} + \nabla_{\vx^*_q} F_{q, \text{out}} + \nabla_{\vx^*_q} \Phi + \Kcopy\nabla_{\vx^*_q} \Copy_q\rangle \\
        & \leq \langle\vx_q'^* - \vx_q^*, \nabla_{\vx^*_q} F_{q, \text{in}} + \nabla_{\vx^*_q} F_{q, \text{out}} + \nabla_{\vx^*_q} \Phi\rangle -\Kcopy \norm{\vx^*_q - \vx'^*_q}_\infty + \epsilon \\
        & \leq mn \norm{\vx^*_q - \vx'^*_q}_\infty (\norm{\nabla_{\vx^*_q} F_{q, \text{in}}}_\infty + \norm{\nabla_{\vx^*_q} F_{q, \text{out}}}_\infty+ \norm{\nabla_{\vx^*_q} \Phi}_\infty) -\Kcopy \norm{\vx^*_q - \vx'^*_q}_\infty + \epsilon \\
        & \leq \left(mn(2\kappa K + 2 n  + m\delta) - \Kcopy\right) \norm{\vx^*_q - \vx'^*_q}_\infty + \epsilon.
    \end{align*}
    Combining with \eqref{eq:total_value_after}, we have
    \begin{align*}
        \left(mn(2\kappa K + 2 n  + m\delta) - \Kcopy\right) \norm{\vx^*_q - \vx'^*_q}_\infty + \epsilon \geq -\epsilon.
    \end{align*}
    The inequality above gives
    \begin{align*}
    \norm{\vx^*_q - \vx'^*_q}_\infty \leq \frac{2\epsilon}{ \Kcopy - mn(2\kappa K + 2 n  + m\delta)} \leq \frac{4\epsilon}{\Kcopy},
    \end{align*}    
    where in the last step we use the fact that $\Kcopy = 4mn\kappa (K + n + m\delta)/\epsilon \geq 2 mn(2\kappa K + 2n + m\delta)$ for any $\epsilon \leq 1$. The proof for $\norm{\vy^*_q - \vy'^*_q}_\infty$ follows similarly.
\end{proof}

We have the following corollary from the previous lemma.

\begin{corollary} \label{cor:value_copy}
    For every $\epsilon$-first-order stationary point and for any node $q \in V,$ it is the case that 
    \begin{align*}
        \left|V_q - \norm{\vx^*_q - \vy^*_q}_2^2 \right| \leq \frac{8mn\epsilon}{\Kcopy}.
    \end{align*}
\end{corollary}

\begin{proof}
    From \Cref{lem:copy_consist}, for every $r \in [m]$ and $i \in [n],$ we have
    \begin{align} \label{eq:copy_close}
        \left|(x^*_{q, r, i} - y^*_{q, r, i}) - (x'^*_{q,r,i} - y'^*_{q,r,i})\right| \leq \frac{8\epsilon}{\Kcopy}.
    \end{align}
    Thus, from the definition of $V_q$ in~\eqref{align:Vq}, we conclude that
    \begin{align*}
        \left|V_q - \norm{\vx^*_q - \vy^*_q}_2^2\right| \leq \frac{8mn\epsilon}{\Kcopy}.
    \end{align*}
\end{proof}

\subsection{Well-guessed copies}

For every node $q \in V,$ $\nabla_{\vx_q} F$ consists of four parts: $\nabla_{\vx_q} F_{q, \text{in}}$, $\nabla_{\vx_q} F_{q, \text{out}}$, $\nabla_{\vx_q} \Phi$, and the gradient from the $\Copy$ gadget. Our next goal is to bound the second and third part. 

For any node $q \in V,$ we have
\begin{align*}
    \nabla_{\vx_q} F_{q, \text{in}} = K \left( \sum_{(q, v, w) \in G_{\Nor}} \alpha_1^{q,v,w} + \sum_{(u, q, w) \in G_{\Nor}} \alpha^{u, q, w}_2 + \sum_{(q, v, w) \in G_{\Purify}} (\beta^{q, v, w}_1 + \gamma^{q, v, w}_1) \right) (\vx_q' - \vy_q').
\end{align*}
For the regularization term, we have
\begin{align*}
    \nabla_{\vx_{q,r}} \Phi = M_{r} (\vx'_{q,r} - \vy'_{q,r}).
\end{align*}
We define 
\begin{align*}
    \Delta_q \defeq K \left( \sum_{(q, v, w) \in G_{\Nor}} \alpha_1^{q,v,w} + \sum_{(u, q, w) \in G_{\Nor}} \alpha^{u, q, w}_2 + \sum_{(q, v, w) \in G_{\Purify}} (\beta^{q, v, w}_1 + \gamma^{q, v, w}_1) \right).
\end{align*}
Thus, for every $q \in V$, $r \in [m]$, $i \in [n]$, $\nabla_{\vx_{q,r,i}} F_{q, \text{in}} + \nabla_{\vx_{q,r,i}} \Phi$ = $(\Delta_q + M_{r})(x'_{q,r,i} - y'_{q,r,i}).$ Moreover, due to the symmetry between $\vx'_q - \vy'_q$ and $\vx_q - \vy_q$, we also have $\nabla_{\vx'_{q,r,i}} F_{q, \text{in}} + \nabla_{\vx'_{q,r,i}} \Phi$ = $(\Delta_q + M_{r})(x_{q,r,i} - y_{q,r,i}).$
We proceed to show the following lemma.

\begin{lemma}\label{lem:bound_on_copy}
    For any $\epsilon$-first-order stationary point, for any $q \in V, r \in [m]$, and $i \in [n]$, if $\Delta_q + M_{r} \neq 0$, it holds that
    \begin{align*}
        \left| x^*_{q,r,i} - y^*_{q,r,i}\right| \leq \frac{2ns_q}{\left|\Delta_q + M_{r}\right|} + 3\sqrt{\frac{\epsilon}{|\Delta_q + M_{r}|}}.
    \end{align*}
\end{lemma}
\begin{proof}
We consider two cases.
    \begin{itemize}
        \item For the case where $\Delta_q + M_{r} > 0,$ we consider the deviation from $x^*_{q, r, i}$ to $y^*_{q,r,i}.$ From the definition of $\epsilon$-first-order stationary point, we have
        \begin{align} \label{eq:x_deviate_vi}
            \langle y^*_{q,r,i} - x^*_{q, r, i}, \nabla_{x_{q,r,i}} F_{q, \text{in}} + \nabla_{x_{q,r,i}} F_{q, \text{out}} + \nabla_{x_{q,r,i}} \Phi + \Kcopy(z'^*_{q,r,i,+} - z'^*_{q,r,i,-}) \rangle \geq -\epsilon.
        \end{align}
        Similarly, if we consider the deviation from $x'^*_{q,r,i}$ to $y'^*_{q,r,i}$, we have
        \begin{align} \label{eq:xp_deviate_vi}
            \langle y'^*_{q,r,i} - x'^*_{q, r, i}, \nabla_{x'_{q,r,i}} F_{q, \text{in}} + \nabla_{x'_{q,r,i}} F_{q, \text{out}} + \nabla_{x'_{q,r,i}} \Phi + \Kcopy(-z'^*_{q,r,i,+} + z'^*_{q,r,i,-}) \rangle \geq -\epsilon.
        \end{align}
        We sum up \eqref{eq:x_deviate_vi} and \eqref{eq:xp_deviate_vi} and analyze each term individually. First,
        \begin{align*}
            & \langle y^*_{q,r,i} - x^*_{q, r, i}, \nabla_{x_{q,r,i}} F_{q, \text{out}}\rangle + \langle y'^*_{q,r,i} - x'^*_{q, r, i}, \nabla_{x'_{q,r,i}} F_{q, \text{out}}\rangle \\
            & = (y^*_{q,r,i} - x^*_{q, r, i}) \left(s_q \nabla_{x_{q,r,i}} H_q\right) + (y'^*_{q,r,i} - x'^*_{q, r, i})\left(s_q \nabla_{x'_{q,r,i}} H_q\right) \\
            & = s_q\left((y^*_{q,r,i} - x^*_{q, r, i})\left(\mat{D}\vx'_{q,r} + \vc\right)_i + (y'^*_{q,r,i} - x'^*_{q,r,i})\left(\sum_{j \in [n]} \mat{D}_{j, i}(x_{q,r,j} - y_{q,r,j}) \right)\right) \\
            & \leq s_q \left(2n |x^*_{q, r, i} - y^*_{q,r,i}| + n |x'^*_{q, r, i} - y'^*_{q,r,i}| \right),
        \end{align*}
        where in the last step we use the fact that $\norm{\mat{D}\vx_{q,r}' + \vec{c}}_{\infty} \leq 2n$ and $\norm{\mat{D}}_{1} \leq 1.$ From \eqref{eq:copy_close}, we have
        \begin{align*}
            s_q \left(2n |x^*_{q, r, i} - y^*_{q,r,i}| + n |x'^*_{q, r, i} - y'^*_{q,r,i}| \right) &\leq s_q (2n|x^*_{q, r, i} - y^*_{q,r,i}| + n|x^*_{q, r, i} - y^*_{q,r,i}| + \frac{8n\epsilon}{\Kcopy})\\
            &\leq 3ns_q|x^*_{q, r, i} - y^*_{q,r,i}| + \epsilon,
        \end{align*}
        where in the last step we use the fact that $0 \leq s_q \leq 1$ and we set $\epsilon < \frac{1}{2}$ so that $\Kcopy > 8n.$ Furthermore, 
        \begin{align*}
            & \left\langle y^*_{q,r,i} - x^*_{q, r, i}, \nabla_{x_{q,r,i}} F_{q, \text{in}} + \nabla_{x_{q,r,i}} \Phi\right\rangle + \left\langle y'^*_{q,r,i} - x'^*_{q, r, i}, \nabla_{x'_{q,r,i}} F_{q, \text{in}} + \nabla_{x'_{q,r,i}} \Phi\right\rangle \\
            & = \left\langle y^*_{q,r,i} - x^*_{q, r, i}, (\Delta_q + M_{r}) (x'^*_{q,r,i} - y'^*_{q,r,i})\right\rangle + \left\langle y'^*_{q,r,i} - x'^*_{q, r, i}, (\Delta_q + M_{r})(x^*_{q,r,i} - y^*_{q,r,i}) \right\rangle\\
            & = -2\left(\Delta_q + M_{r}\right)\left(x^*_{q,r,i} - y^*_{q,r,i}\right)\left(x'^*_{q,r,i} - y'^*_{q,r,i}\right) \\
            & \leq -2\left(\Delta_q + M_{r}\right) \left(x^*_{q,r,i} - y^*_{q,r,i}\right)^2 + \left(\Delta_q + M_{r}\right) \frac{8\epsilon}{\Kcopy},
        \end{align*}
        where the last step follows from \eqref{eq:copy_close}. 
        
        We have $0 \leq \Delta_q \leq \kappa K$ and $|M_{r}| \leq m\delta$. Since we set $\Kcopy = 4mn\kappa (K + n + m\delta)/\epsilon$, by choosing any $\epsilon < \frac{1}{2},$ we also have $(\Delta_q + M_{r})\frac{8\epsilon}{\Kcopy} \leq \epsilon.$ Thus, we conclude that
        \begin{align*}
            &\left\langle y^*_{q,r,i} - x^*_{q, r, i}, \nabla_{x_{q,r,i}} F_{q, \text{in}} + \nabla_{x_{q,r,i}} \Phi\right\rangle + \left\langle y'^*_{q,r,i} - x'^*_{q, r, i}, \nabla_{x'_{q,r,i}} F_{q, \text{in}} + \nabla_{x'_{q,r,i}} \Phi\right\rangle \\
            & \leq  -2\left(\Delta_q + M_{r}\right) \left(x^*_{q,r,i} - y^*_{q,r,i}\right)^2 + \epsilon.
        \end{align*}
        Finally, for the $\Copy$ gadget, we have
        \begin{align*}
            &\left\langle y^*_{q,r,i} - x^*_{q, r, i}, \Kcopy\left(z'^*_{q,r,i,+} - z'^*_{q,r,i,-}\right)\right\rangle + \left\langle y'^*_{q,r,i} - x'^*_{q, r, i}, \Kcopy\left(- z'^*_{q,r,i,+} + z'^*_{q,r,i,-}\right)\right\rangle \\
            & = \Kcopy \left(z'^*_{q,r,i,+} \left(x'^*_{q, r, i} - y'^*_{q, r, i} - x^*_{q, r, i} + y^*_{q, r, i}\right) + z'^*_{q,r,i,-} \left(x^*_{q, r, i} - y^*_{q, r, i} - x'^*_{q, r, i} + y'^*_{q, r, i}\right)\right)\\
            & \leq 8 \epsilon,
        \end{align*}
        where the last step follows from \eqref{eq:copy_close} and the fact that $\left|z'^*_{q,r,i,+} - z'^*_{q,r,i,-}\right| \leq 1$.
        Thus, summing up \eqref{eq:x_deviate_vi} and \eqref{eq:xp_deviate_vi}, we get
        \begin{align*}
            -2\left(\Delta_q + M_{r}\right)\left(x^*_{q,r,i} - y^*_{q,r,i}\right)^2 + 3n s_q \left |x^*_{q,r,i} - y^*_{q,r,i} \right| + 10 \epsilon \geq -2\epsilon.
        \end{align*}
        Using the assumption that $\Delta_q + M_{r} > 0$ and solving for the above quadratic, we get
        \begin{align*}
            \left|x^*_{q,r,i} - y^*_{q,r,i}\right| & \leq \frac{3n s_q + \sqrt{9n^2 s_q^2 + 96(\Delta_q + M_{r}) \epsilon}}{4\left(\Delta_q + M_{r}\right)} \\
            & \leq \frac{3ns_q + \sqrt{9n^2 s_q^2} + \sqrt{96(\Delta_q + M_{r})\epsilon}}{4\left(\Delta_q + M_{r}\right)} \\
            & \leq \frac{2ns_q}{\Delta_q + M_{r}} + 3\sqrt{\frac{\epsilon}{\Delta_q + M_{r}}}.
        \end{align*}
        \item For the case where $\Delta_q + M_{r} < 0,$ consider a deviation from $y^*_{q,r,i}$ to $x^*_{q,r,i}$ and from $y'^*_{q,r,i}$ to $x'^*_{q,r,i},$ we get
        \begin{align} \label{eq:y_deviate_vi}
            \langle x^*_{q,r,i} - y^*_{q, r, i}, \nabla_{y_{q,r,i}} F_{q, \text{in}} + \nabla_{y_{q,r,i}} F_{q, \text{out}} + \nabla_{y_{q,r,i}} \Phi + \Kcopy(z^*_{q,r,i,+} - z^*_{q,r,i,-}) \rangle \leq \epsilon
        \end{align}
        and 
        \begin{align} \label{eq:yp_deviate_vi}
            \langle x'^*_{q,r,i} - y'^*_{q, r, i}, \nabla_{y'_{q,r,i}} F_{q, \text{in}} + \nabla_{y'_{q,r,i}} F_{q, \text{out}} + \nabla_{y'_{q,r,i}} \Phi + \Kcopy(-z^*_{q,r,i,+} + z^*_{q,r,i,-}) \rangle \leq \epsilon.
        \end{align}
        Summing up \eqref{eq:y_deviate_vi} and \eqref{eq:yp_deviate_vi}, all other terms follow similarly as the previous case except for $F_{q, \text{out}}$. For $F_{q, \text{out}}$, we have
        \begin{align*}
            & \langle x^*_{q,r,i} - y^*_{q, r, i}, \nabla_{y_{q,r,i}} F_{q, \text{out}}\rangle + \langle x'^*_{q,r,i} - y'^*_{q, r, i}, \nabla_{y'_{q,r,i}} F_{q, \text{out}}\rangle \\
            & = s_q(x^*_{q,r,i} - y^*_{q, r, i})\left(-\mat{D}\vx'_{q,r} - \vc\right)_i \\
            & \geq - 2ns_q \left|x^*_{q,r,i} - y^*_{q, r, i}\right|.
        \end{align*}
        Thus, combining \eqref{eq:y_deviate_vi} and \eqref{eq:yp_deviate_vi}, we get
        \begin{align}
            2 \left|\Delta_q + M_{r}\right| (x^*_{q,r,i} - y^*_{q, r, i})^2 - 2ns_q \left|x^*_{q,r,i} - y^*_{q, r, i}\right| - 9\epsilon \leq 2\epsilon.
        \end{align}
    Solving for the quadratic,
        \begin{align*}
            \left|x^*_{q,r,i} - y^*_{q, r, i}\right| \leq \frac{ns_q}{\left|\Delta_q + M_{r}\right|} + 3\sqrt{\frac{\epsilon}{|\Delta_q + M_{r}|}} \leq \frac{2ns_q}{\left|\Delta_q + M_{r}\right|} + 3\sqrt{\frac{\epsilon}{|\Delta_q + M_{r}|}}.
        \end{align*}
    \end{itemize}
    Thus, we conclude that if $\Delta_q + M_{r} \neq 0$, regardless of the sign of $\Delta_q + M_{r},$ we have $\left| x^*_{q,r,i} - y^*_{q,r,i}\right| \leq \frac{2ns_q}{\left|\Delta_q + M_{r}\right|} + 3\sqrt{\frac{\epsilon}{|\Delta_q + M_{r}|}}$. This completes the proof.
\end{proof}

We proceed to show the following technical lemma concerning well-guessed copies.
\begin{lemma} \label{lem:well_guess}
    For every node $q \in V$ and for any constant $ \rho < 1,$ there are at least $\frac{\rho}{2\delta}$ copies $r \in [m]$ that satisfy $|\Delta_q + M_{r}| \leq \rho$. Moreover, there are at most two copies that satisfy $|\Delta_q + M_{r}| < \delta.$
\end{lemma}

\begin{proof}
    We set $h_q = \frac{m}{2} - \frac{\Delta_q}{\delta}$. From the definition of $M_{r}$ in \eqref{eq:M_def}, we can write $\Delta_q + M_{r} = \delta(r - h_q).$ Since $K = 10^4\kappa^2 n^9 \log (m)$ and $\delta = \frac{1}{\kappa^2 n^6}$, for sufficiently large $n$, we have 
    \begin{align*}
        \frac{m\delta}{4} \geq \frac{n^{10}}{4 \kappa^2} \geq \kappa K \geq \Delta_q,
    \end{align*}
    where we use the fact that $m \geq n^{16}$ and $0 \leq \Delta_q \leq \kappa K.$ Thus, we have $0 \leq \Delta_q \leq \frac{m\delta}{4}$ which implies $\frac{m}{4} \leq h_q \leq \frac{m}{2}.$
    Since $h_q \in [\frac{m}{4}, \frac{m}{2}]$ and $1 < \frac{\rho}{\delta} = \rho \kappa^2 n^6 < m$, there are at least $\lfloor \frac{\rho}{\delta} \rfloor \geq \frac{\rho}{2\delta}$ copies $r \in [m]$ such that $|r - h_q| \leq \frac{\rho}{\delta}.$ Given that $|r - h_q| \leq \frac{\rho}{\delta}$ implies $|\Delta_q + M_{r}| \leq \rho,$ this proves the first part of the lemma.

    To prove the second part, since having copies $r \in [m]$ such that $|\Delta_q + M_{r}| < \delta$ is equivalent to having copies satisfying $|r - h_q| < 1.$ Because $r$ is an integer, there are at most two integers that satisfy this condition.
\end{proof}

From \Cref{lem:well_guess}, we have the following corollary.

\begin{corollary} \label{cor:well_sum}
    For any node $q \in V,$ consider all the copies $r \in [m]$ such that $|\Delta_q + M_{r}| \geq \delta$, it holds that
    \begin{align*}
        \sum_{r: |\Delta_q + M_{r}| \geq \delta} \frac{1}{|\Delta_q + M_{r}|} \leq \frac{4 \log(m)}{\delta} \quad \text{and} \quad  \sum_{r: |\Delta_q + M_{r}| \geq \delta} \frac{1}{|\Delta_q + M_{r}|^2} \leq \frac{4}{\delta^2}.
    \end{align*}
\end{corollary}

\begin{proof}
    From \Cref{lem:well_guess}, we can group all the copies $r \in [m]$ such that $k \leq |r - h_q| < k+1$ for any integer $k \in [1,m].$ For each interval $[k, k+1),$ there are at most 2 copies satisfying this condition. Thus,
    \begin{align*}
        \sum_{r: |\Delta_q + M_{r}| \geq \delta} \frac{1}{|\Delta_q + M_{r}|} \leq \frac{2}{\delta} \sum_{k = 1}^m \frac{1}{k} \leq \frac{4 \log(m)}{\delta}.
    \end{align*}
    Similarly, we can show that
    \begin{align*}
        \sum_{r: |\Delta_q + M_{r}| \geq \delta} \frac{1}{|\Delta_q + M_{r}|^2} \leq \frac{2}{\delta^2} \sum_{k = 1}^m \frac{1}{k^2} \leq \frac{4}{\delta^2}.
    \end{align*}
\end{proof}

\subsection{Gadget consistency}

Having established the characterization of well-guessed copies in the previous subsection, we move on to show the consistency of the $\Nor$ and $\Purify$ gadgets we design.

\begin{lemma}\label{lem:consist_inactive}
    Let $\epsilon < n^{-8}$. Given any $\epsilon$-first-order stationary point, for any node $q \in V,$ if $s_q \leq \epsilon$, then it is the case that $V_q \leq 3n.$
\end{lemma}
\begin{proof}
    From \Cref{lem:bound_on_copy}, we have
    \begin{align*}
        (x^*_{q,r,i} - y^*_{q,r,i})^2 \leq \frac{8n^2 \epsilon^2}{|\Delta_q + M_{r}|^2} + \frac{18\epsilon}{|\Delta_q + M_{r}|}.
    \end{align*}
    From \Cref{lem:well_guess} and \Cref{cor:well_sum}, it holds that
    \begin{align*}
        \norm{\vx^*_q - \vy^*_q}_2^2 & \leq  \sum_{r: |\Delta_q + M_{r}|<\delta} \norm{\vx_{q,r}^* - \vy^*_{q,r}}_2^2  +  \sum_{r: |\Delta_q + M_{r}|\geq \delta} \norm{\vx_{q,r}^* - \vy^*_{q,r}}_2^2 \\
        & \leq 2n + n\left(\frac{32n^2 \epsilon^2}{\delta^2} + \frac{72\log(m)\epsilon}{\delta} \right) \leq 3n,
    \end{align*}
    where in the last step we use the fact that $\epsilon < n^{-8}$ so $\frac{32n^3\epsilon^2}{\delta^2} \leq \frac{n}{2}$ and $\frac{72\log(m) \epsilon}{\delta} < \frac{n}{2}$ for sufficiently large $n$. The proof is complete.
\end{proof}

We complement \Cref{lem:consist_inactive} with the following lemma.

\begin{lemma} \label{lem:consist_active}
    Given any $\epsilon$-first-order stationary point, for any node $q \in V,$ if $s_q \geq 1 - \epsilon$, and there is no copy $r \in [m]$ such that $\vx'^*_{q, r}$ is a $\rho$-solution to \linVI, then $V_q \geq 3n + 1$.
\end{lemma}

\begin{proof}
    Let $r$ be a well-guessed copy such that $|\Delta_q + M_{r}| \leq \rho' = \frac{\rho}{24}.$ We aim to show that $\norm{\vx^*_{q,r} - \vy^*_{q,r}}_\infty \geq \frac{\rho}{6n}.$

    Suppose, for the sake of contradiction, that $\norm{\vx^*_{q,r} - \vy^*_{q,r}}_\infty \leq \frac{\rho}{6n}.$ For any $i \in [n]$ and any $x \in [0, 1]$, we consider the deviations from $y^*_{q,r,i}$ to $x$ and from $y'^*_{q,r,i}$ to $x.$ Since $y^*_{q,r,i}$ and $y'^*_{q,r,i}$ are part of an $\epsilon$-first-order stationary point, we have
    \begin{align*}
        \langle x - y^*_{q,r,i}, \nabla_{y_{q,r,i}} F\rangle \leq \epsilon \quad \text{and} \quad \langle x - y'^*_{q,r,i}, \nabla_{y'_{q,r,i}} F\rangle \leq \epsilon.
    \end{align*}
    Summing up these two inequalities, we get
    \begin{align*}
        & \langle x - y^*_{q,r,i}, \nabla_{y_{q,r,i}} F\rangle + \langle x - y'^*_{q,r,i}, \nabla_{y'_{q,r,i}} F\rangle \\
        & = \langle x - y^*_{q,r,i}, \nabla_{y_{q,r,i}} F_{q, \text{in}} + \nabla_{y_{q,r,i}} F_{q, \text{out}} + \nabla_{y_{q,r,i}} \Phi + \Kcopy(z^*_{q,r,i,+} - z^*_{q,r,i,-})\rangle \\
        & \quad + \langle x - y'^*_{q,r,i}, \nabla_{y'_{q,r,i}} F_{q, \text{in}} + \nabla_{y'_{q,r,i}} F_{q, \text{out}} + \nabla_{y'_{q,r,i}} \Phi + \Kcopy(-z^*_{q,r,i,+} + z^*_{q,r,i,-})\rangle
    \end{align*}
    We analyze each term individually. First, for the terms involving $F_{q, \text{in}}$ and $\Phi$,
    \begin{align*}
        & \langle x - y^*_{q,r,i}, \nabla_{y_{q,r,i}} F_{q, \text{in}} + \nabla_{y_{q,r,i}} \Phi \rangle + \langle x - y'^*_{q,r,i}, \nabla_{y'_{q,r,i}} F_{q, \text{in}} + \nabla_{y'_{q,r,i}} \Phi\rangle \\
        & = - (x - y^*_{q,r,i})(\Delta_q + M_{r})(x'^*_{q,r,i} - y'^*_{q,r,i}) - (x - y'^*_{q,r,i})(\Delta_q + M_{r})(x^*_{q,r,i} - y^*_{q,r,i})\\
        & \geq -2|\Delta_q + M_{r}| \geq -2\rho',
    \end{align*}
    where in the last step we use the fact that $r$ is a well-guessed copy. For the $\Copy$ gadget term, we have
    \begin{align*}
        & \langle x - y^*_{q,r,i}, \Kcopy (z^*_{q,r,i,+} - z^*_{q,r,i, -}) + \langle x - y'^*_{q,r,i}, \Kcopy (-z^*_{q,r,i,+} + z^*_{q,r,i, -})\rangle \\
        & = (y'^*_{q,r,i} - y^*_{q,r,i}) \Kcopy (z^*_{q,r,i,+} - z^*_{q,r,i,-}) \geq -4\epsilon,
    \end{align*}
    where the last step follows from \Cref{lem:copy_consist}.
    Finally, for the term involving $F_{q, \text{out}},$ we have
    \begin{align*}
        & \langle x - y^*_{q,r,i}, \nabla_{y_{q,r,i}} F_{q, \text{out}}\rangle + \langle x - y'^*_{q,r,i}, \nabla_{y'_{q,r,i}} F_{q, \text{out}}\rangle\\
        & = -s_q(x - y^*_{q,r,i})(\mat{D}\vx'_{q,r} + \vc)_i \\
        & \geq -s_q(x - x'^*_{q,r,i})(\mat{D}\vx'_{q,r} + \vc)_i - s_q (x'^*_{q,r,i} - y^*_{q,r,i})(\mat{D}\vx'_{q,r} + \vc)_i\\
        & \geq -s_q(x - x'^*_{q,r,i})(\mat{D}\vx'_{q,r} + \vc)_i - 2n \norm{\vx'^*_{q, r} - \vy^*_{q,r}}_\infty\\
        & \geq -s_q(x - x'^*_{q,r,i})(\mat{D}\vx'_{q,r} + \vc)_i - 2n \norm{\vx'^*_{q, r} - \vx^*_{q,r}}_\infty - 2n \frac{4\epsilon}{\Kcopy},
    \end{align*}
    where the last step holds because of \Cref{lem:copy_consist}.
    Thus, for all $x \in [0, 1]$,
    \begin{align*}
        & -2\rho' - 4\epsilon  -s_q(x - x'^*_{q,r,i})(\mat{D}\vx'_{q,r} + \vc)_i - 2n \norm{\vx'^*_{q, r} - \vx^*_{q,r}}_\infty - 2n \frac{4\epsilon}{\Kcopy} \leq 2\epsilon.
    \end{align*}
    This further gives
    \begin{align*}
         s_q (x - x'^*_{q,r,i})(\mat{D}\vx'_{q,r} + \vc)_i \geq -2\rho' -\frac{\rho}{3} - 7 \epsilon.
    \end{align*}
    Using $s_q \geq 1 - \epsilon > \frac{1}{2}$ and set $\rho' = \frac{\rho}{24}, \epsilon < 0.01\rho$, for all $x \in [0, 1],$ we have
    \begin{align*}
        (x - x'^*_{q,r,i})(\mat{D}\vx'_{q,r} + \vc)_i \geq -4\rho' - \frac{2\rho}{3} -14 \epsilon > -\rho.
    \end{align*}
    This contradicts the fact that $\vx'_{q,r}$ is not a $\rho$-solution to \linVI. Thus, for all well-guessed copies $r \in [m],$ we have $\norm{\vx^*_{q,r} - \vy^*_{q,r}}_{\infty} \geq \frac{\rho}{6n}.$ From \Cref{lem:well_guess}, there are at least $\frac{\rho'}{2\delta}$ well-guessed copies. Therefore,
    \begin{align*}
        \norm{\vx_q - \vy_q}^2 \geq \frac{\rho'}{2\delta} \cdot \frac{\rho^2}{36n^2}.
    \end{align*}
    Since we set the parameter $\delta = \kappa^{-2}n^{-6}$ and $\rho, \rho'$ are constants, we have $\norm{\vx_q - \vy_q}_2^2 \geq 3n + 2.$ Finally, from \Cref{cor:value_copy} and the fact that $\Kcopy > 8mn\epsilon,$ we have
    \begin{align*}
        V_q \geq 3n + 2 - \frac{8mn\epsilon}{\Kcopy} \geq 3n + 1.
    \end{align*}
    The proof is complete.
\end{proof}

\subsection{Gate consistency}

So far, we have shown that $V_q$ satisfies the threshold property given the value of each node $s_q.$ We now prove the other direction. Namely, given the input values, we show that the value $s_q$ correctly captures the output value. We start with the following technical lemma on bounding the linking term $H_q  = \sum_{r \in [m]} \langle \mat{D}\vx'_{q,r} + \vc, \vx_{q,r} - \vy_{q,r}\rangle.$ 

\begin{lemma} \label{lem:bound_on_link}
    For any $\epsilon$-first-order stationary point, for every node $q \in V$, it holds that
    $ |H_q| \leq \frac{K}{100}.$  
\end{lemma}
\begin{proof}
    From \Cref{lem:bound_on_copy} and \Cref{cor:well_sum}, for any $\epsilon$-first-order stationary point, we have
    \begin{align*}
        \norm{\vx_q - \vy_q}_1 & \leq 2n + \sum_{r: |\Delta_q + M_{r}| \geq \delta} \sum_{i \in [n]} \left(\frac{2ns_q}{|\Delta_q + M_{r}|} + 3 \sqrt{\frac{\epsilon}{|\Delta_q + M_{r}|}}\right) \\
        & \leq 2n + n\left( \frac{8n\log(m)}{\delta} + 3m\sqrt{\frac{\epsilon}{\delta}}\right) \\
        & \leq 2n + 8\kappa^2n^8\log(m) + 3 \leq 10 \kappa^2 n^8 \log(m),
    \end{align*}
    where we use the fact that $\delta = \kappa^{-2}n^{-6}$ and $\epsilon < n^{-39} \leq \kappa^{-2}n^{-6}m^{-2}.$ 
    Since $\norm{\mat{D}\vx'_{q,r} + \vc}_{\infty} \leq 2n$, we have
    \begin{align*}
        |H_q| \leq 2n \norm{\vx_q - \vy_q}_1 \leq 20 \kappa^2 n^9 \log(m) \leq \frac{K}{100},
    \end{align*}
    where in the last step we use the fact that $K = 10^4 \kappa^2 n^9 \log(m)$.
\end{proof}

The preceding lemma, together with our choice of $K$, plays an important role in proving gate consistency. We first establish the correctness of the $\Nor$ gate, where
\begin{align} \label{eq:nor_gate}
     \Nor_{u, v, w} \defeq & K\left(\alpha^{u ,v, w}_{1} \left(V_u -3n - \frac{1}{2} \right) + \alpha^{u,v,w}_2 \left(V_v -3n - \frac{1}{2} \right)\right) +\alpha_3^{u, v, w} H_w.
\end{align}
In particular, we have the following lemma.
\begin{lemma}\label{lem:nor_correct}
    For any $\epsilon$-first-order stationary point, consider a $\Nor$ gate with $(u, v, w)$ where $u, v$ are the input nodes and $w$ is the output node. We have
    \begin{itemize}
        \item If $V_u \leq 3n$ and $V_v \leq 3n,$ then $s_w \geq 1-\epsilon.$
        \item If $V_u \geq 3n + 1$ or $V_v \geq 3n + 1,$ then $s_w < \epsilon.$
    \end{itemize}
\end{lemma}

\begin{proof}
    If $V_u \leq 3n$ and $V_v \leq 3n$, the payoffs of actions $\alpha_{1}^{u,v,w}$ and $\alpha_2^{u,v,w}$ in \eqref{eq:nor_gate} are $K(V_u - 3n - \frac{1}{2}) \leq -\frac{K}{2}$ and $K(V_v - 3n - \frac{1}{2}) \leq -\frac{K}{2}$ respectively. Meanwhile, from \Cref{lem:bound_on_link}, the value of \eqref{eq:nor_gate} when $\valpha^{u, v, w}$ chooses $\alpha^{u, v, w}_3$ is bounded below by $-\frac{K}{100}.$ Let $U(\valpha^{u,v,w})$ denote the payoff for ${\valpha^{u ,v, w}}$ in the $\Nor$ gadget, from the $\epsilon$-first-order stability condition of $\valpha^{u,v,w, *}$, we have
    \begin{align*}
        U(\alpha^{u,v,w}_3) - U(\valpha^{u, v, w, *}) \leq \epsilon.
    \end{align*}
    Substituting the lower bound for $U(\alpha_3^{u, v, w}),$ we get
    \begin{align*}
        & -\frac{K}{100} - U(\valpha^{u, v, w, *}) \leq \epsilon.
    \end{align*}
    Since $U(\alpha_1^{u, v, w,*}) \leq -\frac{K}{2},$ it holds that
    \begin{align*}
        &\left(-\frac{K}{100} + \frac{K}{2}\right) \alpha^{u ,v ,w, *}_1  \leq \epsilon.
    \end{align*}
    Thus, we conclude that $\alpha^{u,v,w,*}_1 \leq \frac{4\epsilon}{K} \leq \frac{\epsilon}{2}$. Similarly, we can show $\alpha_2^{u ,v, w, *} \leq \frac{\epsilon}{2}$ and therefore $s_w = \alpha_3^{u, v,w, *} \geq 1 - \epsilon.$

    On the other hand, without loss of generality, assume that $V_u \geq 3n + 1$,
    \begin{align*}
        & U(\alpha_1^{u ,v, w}) - U(\valpha^{u ,v, w, *}) \leq \epsilon.
    \end{align*}
    Substituting the bounds for $U(\alpha_1^{u, v, w})$ and $U(\alpha_3^{u, v, w, *}),$ and by following similar analysis from the case above, we get
    \begin{align*}
        \left(\frac{K}{2} - \frac{K}{100}\right) \alpha_{3}^{u ,v, w,*} \leq \epsilon.
    \end{align*}
    We conclude that in this case, $s_w = \alpha^{u ,v, w, *}_3 \leq \frac{3\epsilon}{K} \leq \epsilon$.
\end{proof}

We continue with the correctness of $\Purify$ gate, where 

\begin{align} \label{eq:puri_gate}
    \Purify_{u,v,w} \defeq & K \beta^{u, v, w}_1 \left( V_u - 3n -\frac{1}{4}\right) + \beta_1^{u, v, w} H_v  + K \gamma^{u, v, w}_1 \left( V_u - 3n -\frac{3}{4}\right) + \gamma_1^{u, v, w} H_w.
\end{align}
\begin{lemma} \label{lem:puri_correct}
    For any $\epsilon$-first-order stationary point, consider a $\Purify$ gate with $(u, v, w)$ where $u$ is the input node and $v, w$ are the output nodes. It holds that
    \begin{itemize}
        \item If $V_u \leq 3n,$ then $s_v \leq \epsilon$ and $s_w \leq \epsilon.$
        \item If $V_u \geq 3n + 1, $ then $s_v \geq 1 - \epsilon$ and $s_w \geq 1 - \epsilon.$
        \item If $3n < V_u < 3n+1,$ then $s_v \geq 1 - \epsilon$ or $s_w \leq \epsilon.$
    \end{itemize}
\end{lemma}

\begin{proof}
    Let $U(\vbeta^{u, v,w})$ and $U(\vgamma^{u, v, w})$ be the payoffs of \eqref{eq:puri_gate} given $\vbeta$ and $\vgamma$ respectively. If $V_u \leq 3n,$ then $U(\beta^{u ,v, w}_1) \leq -\frac{K}{4} + \frac{K}{100} \leq -\frac{K}{5}$, whereas $U(\beta^{u,v,w}_2) = 0.$ Therefore, following a similar argument as in \Cref{lem:nor_correct}, we have
    \begin{align*}
        0 - U(\vbeta^{u ,v, w,*}) \leq \epsilon \quad  \Rightarrow \quad \left(0 + \frac{K}{5}\right) \beta_1^{u ,v ,w, *} \leq \epsilon.
    \end{align*}
    Thus, $s_v = \beta_1^{u, v, w,*} \leq \epsilon.$ A similar argument shows that $s_w \leq \epsilon$ in this case.

    Now suppose $V_u \geq 3n + 1$. We have $U(\beta_1^{u ,v, w, *}) \geq \frac{3K}{4} - \frac{K}{100} \geq \frac{K}{2}$ and $U(\beta^{u, v,w,*}_2) = 0.$ Therefore,
    \begin{align*}
        \frac{K}{2} - U(\vbeta^{u , v, w, *}) \leq \epsilon \quad \Rightarrow \quad \left(\frac{K}{2} - 0\right) \beta_2^{u, v, w, *} \leq \epsilon.
    \end{align*}
    This gives $\beta_2^{u, v, w, *} \leq \epsilon$ and thus $s_v \geq 1 - \epsilon$. Similarly, we can show $s_w \geq 1 - \epsilon$.

    Finally, for the case where $3n < V_u < 3n + 1,$ if $V_u \geq 3n + \frac{1}{2},$ we have $U(\beta_1^{u, v, w, *}) \geq \frac{K}{4} - \frac{K}{100} \geq \frac{K}{5}.$ From a similar argument as in the case where $V_u \geq 3n + 1,$ we can show that $s_v  = \beta_1^{u, v, w, *}\geq 1- \epsilon.$ On the other hand, if $V_u < 3n + \frac{1}{2},$ then $U(\gamma_1^{u ,v, w}) \leq - \frac{K}{4} + \frac{K}{100} \leq -\frac{K}{5}.$ Similarly, when $V_u < 3n + \frac{1}{2}$, it follows that $s_w = \gamma_1^{u, v, w,*} \leq \epsilon.$ The proof is complete.
\end{proof}
We have shown the consistency of the $\Nor$ and $\Purify$ gadgets we design in both directions. Now we are ready to state the main theorem for this section.
\begin{theorem} \label{thm:degree_three_ppad}
    For any $\epsilon < n^{-39}$, finding an $\epsilon$-first-order stationary point for min-max optimization with degree-$3$ multilinear objective is \PPAD-complete.
\end{theorem}

\begin{proof}
    The membership follows, for example, from~\citet{DSZ21}. To prove the hardness, we construct a degree-$3$ multilinear objective function as in \eqref{eq:final_game} and find an $\epsilon$-first-order stationary point of this objective. For every $q \in V, r \in [m]$, and $i \in [n],$ we check if $\vx'_{q, r}$ is a $\rho$-solution of \linVI. There are only $\poly(n)$ coordinates to check. If such $\vx'_{q,r}$ exists, then the \PPAD-hardness follows from the \PPAD-hardness of $\linVI$  (\Cref{theorem:LinVI}). Suppose, on the other hand, that no such $\vx'_{q,r,i}$ exists. For each node $q \in V,$ we define the following assignment function $\lambda(q)$.
    \begin{align*}
        \lambda(q) \defeq \begin{cases}
        0& V_q\le3n,\\
        1& V_q\ge3n+1,\\
        \bot& 3n<V_q<3n+1.
        \end{cases}
    \end{align*}
    From \Cref{lem:consist_inactive,lem:consist_active,lem:nor_correct}, if there is no $x'_{q, r, i}$ which is a $\rho$-solution to \linVI, then for any $\Nor$ gate $(u, v, w)$ in \pureCircuit, the $\Nor$ gate constraints are satisfied with values $\lambda(u), \lambda(v)$, and $\lambda(w)$. Furthermore, from \Cref{lem:consist_inactive,lem:consist_active,lem:puri_correct}, if $x'_{q,r, i}$ is not a $\rho$-solution to $\linVI$ for all $q \in V, r \in [m], i \in [n],$ then for every $\Purify$ gate $(u, v, w)$ in \pureCircuit, the values $\lambda(u), \lambda(v)$, and $\lambda(w)$ satisfy the $\Purify$ gate constraints. Thus, the assignment function $\lambda(\cdot)$ gives a valid solution to \pureCircuit. The \PPAD-hardness  then follows from the \PPAD-hardness for $\pureCircuit$ (\Cref{thm:purecircuit}). The proof is complete. 
\end{proof}

Although the objective function in \Cref{thm:degree_three_ppad} is multilinear, it contains degree-3 terms such as $\alpha^{u, v, w}_1 y_{q,r,i} y'_{q,r,i}$ which capture interactions among three independent maximizers. These terms prevent us from establishing $\PPAD$-hardness for 2 vs.\ 2 team zero-sum games. To overcome this technical difficulty, in the next section we introduce a gadget with a quadratic multilinear objective that simulates the degree-3 objective.

\section{Reducing to a quadratic multilinear function}
\label{sec:quadratic}

In this section, we show that we can approximate the degree-$3$ multilinear objective in \Cref{thm:degree_three_ppad} with a degree-$2$ multilinear objective. In the objective function, the regularization term $\Phi$ and the $\Copy$ gadget have degree two. Therefore, we focus on approximating all degree-3 terms in the objective
\begin{align} \label{eq:degree_3_obj}
    \sum_{(u, v, w) \in G_{\Nor}} \Nor_{u, v, w} + \sum_{(u, v, w) \in G_{\Purify}} \Purify_{u, v, w}
\end{align}
with a degree-2 multilinear function.

Since the function in \eqref{eq:degree_3_obj} is multilinear, we can rewrite all the degree-3 terms as $P = \sum c_k \cdot \tilde{x}_k\tilde{y}_k\tilde{z}_k$,
where $\tilde{x}_k,$ $\tilde{y}_k,$ $\tilde{z}_k \in [0, 1]$ are coordinates of the original strategy variables in \eqref{eq:degree_3_obj} and $c_k$ is the corresponding coefficient of the $k^{th}$ term. We write $\tilde{x}_k\tilde{y}_k\tilde{z}_k$ as
\begin{align*}
    \tilde{x}_k\tilde{y}_k\tilde{z}_k = \frac{1}{6} \left(27 \left(\frac{\tilde{x}_k + \tilde{y}_k + \tilde{z}_k}{3}\right)^3 - 8 \left(\frac{\tilde{x}_k + \tilde{y}_k}{2}\right)^3 - 8 \left(\frac{\tilde{x}_k + \tilde{z}_k}{2}\right)^3 - 8 \left(\frac{\tilde{y}_k + \tilde{z}_k}{2}\right)^3 + \tilde{x}_k^3 + \tilde{y}_k^3 + \tilde{z}_k^3\right).
\end{align*}
Thus, in order to simulate the term $\tilde{x}_k\tilde{y}_k\tilde{z}_k,$ all we need is to construct functions of the form $\phi_k(\tilde{x}_k, \tilde{y}_k,\tilde{z}_k)^3$ where $\phi_k(\tilde{x}_k, \tilde{y}_k,\tilde{z}_k): [0, 1]^3 \to [0,1]$ takes one of the following forms:
\begin{align} \label{eq:phi_choice}
    \frac{\tilde{x}_k + \tilde{y}_k + \tilde{z}_k}{3}, \quad \frac{\tilde{x}_k + \tilde{y}_k}{2}, \quad \frac{\tilde{x}_k + \tilde{z}_k}{2}, \quad \frac{\tilde{y}_k + \tilde{z}_k}{2}, \quad \tilde{x}_k, \quad \tilde{y}_k, \quad \tilde{z}_k.
\end{align}
Note that $\phi_k(\tilde{x}_k,\tilde{y}_k,\tilde{z}_k)$ is always a linear combination of coordinates $\tilde{x}_k, \tilde{y}_k,\tilde{z}_k$ from the original variables in \eqref{eq:degree_3_obj}. We continue with the following definition.
\begin{definition}[Tangent affine function]
Let $f: [0, 1] \to \mathbb{R}$ be a differentiable function, for any point $p \in [0, 1],$ the tangent affine function of $f$ at $p$ is
\begin{align*}
    T^f_p(t) = f(p) + (t - p)f'(p).
\end{align*}
\end{definition}

When $f(t) = ct^3$, we have $T^f_p(t) = cp^3 + 3c(t-p)p^2.$ We proceed with the following technical lemma.
\begin{lemma} [Linear local approximation]\label{lem:grad_approx_err}
    Let $f(t) : [0,1] \to \mathbb{R}$ defined by $f(t) = ct^3$ with $c > 0$ and let $\mathcal{G} \subset [0, 1]$ be a finite grid with mesh $\delta'$ which scales inverse-polynomially in the size of the input. Suppose for any $t \in [0, 1]$, there exists a variable $\vzeta \in \Delta_{|\mathcal{G}|}$ satisfying
    \begin{align*}
        \max_{p \in \mathcal{G}} T^f_{p}(t) - \sum_{p \in \mathcal{G}}\zeta_p T^f_p(t) \leq \epsilon',
    \end{align*}
    then we have
    \begin{align*}
        \left| \sum_{p \in \mathcal{G}} \zeta_p f'(p) - f'(t)\right| \leq \sqrt{18c(c\delta'^2 + \epsilon')}.
    \end{align*}
\end{lemma}

\begin{proof}
    We have
    \begin{align} \label{eq:tan_err}
        f(t) - T^f_p(t) &= ct^3 - cp^3 - 3c(t-p)p^2 = c (t - p)^2 (t + 2p).
    \end{align}
    On the other hand, we have
    \begin{align*}
        \left|f'(p) -f'(t)\right| = 3c|p- t|(p + t).
    \end{align*}
    By using the fact that $0 \leq p + t \leq 2,$ it holds that
    \begin{align*}
        \left|f'(p) - f'(t)\right|^2 \leq 18c\left(f(t) - T^f_p(t)\right).
    \end{align*}
    Since $\mathcal{G} \subset [0, 1]$ is a finite grid with mesh $\delta',$ for any point $t \in [0, 1],$ we can find a point $p_t \in \mathcal{G}$ such that $\left|p_t - t\right| \leq \delta'/2.$ From \eqref{eq:tan_err}, this implies $f(t) - T^f_{p_t}(t) \leq c \delta'^2.$ Therefore, 
    \begin{align*}
        f(t) - \sum_{p \in \mathcal{G}} \zeta_p T^f_p(t) & = f(t) - \max_{p \in \mathcal{G}} T^f_p(t) + \max_{p \in \mathcal{G}} T^f_p(t) -\sum_{p \in \mathcal{G}} \zeta_p T^f_p(t)  \\
        & \leq f(t) - T^f_{p_t}(t) + \epsilon' \\
        & \leq c\delta'^2 + \epsilon'.
    \end{align*}
    Thus, by Jensen's inequality, we conclude that
    \begin{align*}
        \left|\sum_{p \in \mathcal{G}} \zeta_p f'(p) - f'(t)\right| & \leq \sum_{p \in \mathcal{G}} \zeta_p |f'(p) - f'(t)| \\
        & \leq \sqrt{\left(\sum_{p \in \mathcal{G}} \zeta_p |f'(p) - f'(t)|^2 \right)} \\
        & \leq \sqrt{18c\left(\sum_{p \in \mathcal{G}} \zeta_p \left(f(t) - T^f_p(t)\right)\right)} \\
        & \leq\sqrt{18c\left(f(t) - \sum_{p \in \mathcal{G}} \zeta_p T^f_p(t)\right)} \\
        & \leq \sqrt{18c(c\delta'^2 + \epsilon')}.
    \end{align*}
    The proof is complete.
\end{proof}

For the case where $c < 0,$ we have the following corollary.

\begin{corollary}
    Let $f(t) : [0,1] \to \mathbb{R}$ defined by $f(t) = ct^3$ with $c < 0$. Suppose for any $t \in [0, 1]$, there exists a variable $\vzeta \in \Delta_{|\mathcal{G}|}$ satisfying
    \begin{align*}
        \sum_{p \in \mathcal{G}}\zeta_p T^f_p(t) - \min_{p \in \mathcal{G}} T^f_{p}(t) \leq \epsilon',
    \end{align*}
    then we have
    \begin{align*}
        \left| \sum_{p \in \mathcal{G}} \zeta_p f'(p) - f'(t)\right| \leq \sqrt{18|c|(|c|\delta'^2 + \epsilon')}.
    \end{align*}
\end{corollary}
\begin{proof}
    The proof of this corollary  follows similarly from the proof of \Cref{lem:grad_approx_err} with reversed signs.
\end{proof}

Let $\vtau = (\valpha, \vbeta, \vgamma, \vx, \vx', \vy, \vy', \vz, \vz')$ be the concatenation of the original variables in the degree-3 multilinear objective. For every $k$ and every $\phi_k(\tilde{x}_k,\tilde{y}_k,\tilde{z}_k)$ that takes one value in \eqref{eq:phi_choice}, we define $\psi_k(\vtau) = \phi_k(\tilde{x}_k,\tilde{y}_k,\tilde{z}_k)$. Since $\phi_k(\tilde{x}_k,\tilde{y}_k,\tilde{z}_k)$ is always a linear combination of coordinates $\tilde{x}_k, \tilde{y}_k,\tilde{z}_k$, $\psi_k(\cdot)$ is an affine function of $\vtau$. \Cref{lem:grad_approx_err} suggests that if one can find $\vzeta \in \Delta_{|\mathcal{G}|}$ such that $\sum_{p} \zeta_p T^f_p(t)$ is close to $\max_{p} T^f_p(t),$ then $\sum_p\zeta_p f'(p)$ is a good approximation of $f'(t).$ For each coordinate combination $c_k \phi_k(\tilde{x}_k, \tilde{y}_k, \tilde{z}_k)$ with positive coefficient $c_k > 0,$ this allows us to introduce a maximizer variable $\vzeta_{k}^+ \in \Delta_{|\mathcal{G}|}$ as in \Cref{lem:grad_approx_err} and define the approximation gadget
\begin{align} \label{eq:approx_gadget}
    G_{k, \psi_{k}}^+(\vtau, \vzeta^+_{k}) \defeq \sum_{p \in \mathcal{G}} \zeta_{k, p}^+ T^f_p(\psi_k(\vtau)),
\end{align}
where $f(t): [0,1] \to \mathbb{R} $ is given as $f(t) = c_k t^3.$ On the other hand, with negative coefficient $c_k < 0$, we introduce a minimizer variable $\vzeta_k^- \in \Delta_{|\mathcal{G}|}$ and the following approximation gadget 
\begin{align} \label{eq:approx_gadget_n}
    G_{k, \psi_{k}}^-(\vtau, \vzeta_{k}^-) \defeq \sum_{p \in \mathcal{G}} \zeta_{k, p}^- T^f_p(\psi_k(\vtau)).
\end{align}
We show the following lemma with the above approximation gadget.
\begin{lemma} \label{lem:approx_gadget}
    Let the objective function be the approximation gadgets defined in \eqref{eq:approx_gadget} and \eqref{eq:approx_gadget_n}, at any $\epsilon'$-first-order stationary point, if $c_k > 0,$ we have
    \begin{align*}
        \norm{\nabla_{\vtau} G_{k, \psi_k}^+(\vtau^*, \vzeta_k^{+,*}) - \nabla_{\vtau} \left(c_k \psi_k(\vtau^*)^3\right)}_\infty \leq \sqrt{18c_k(c_k \delta'^2 + \epsilon')}.
    \end{align*}
    If $c_k < 0,$ we have
    \begin{align*}
        \norm{\nabla_{\vtau} G_{k, \psi_k}^-(\vtau^*, \vzeta_k^{-,*}) - \nabla_{\vtau} \left(c_k \psi_k(\vtau^*)^3\right)}_\infty \leq \sqrt{18|c_k|(|c_k| \delta'^2 + \epsilon')}.
    \end{align*}
\end{lemma}

\begin{proof}
    We prove the case with positive coefficient $c_k > 0$; the other case follows similarly. Since the maximizer $\vzeta_k^+$ only appears in the approximation gadget term, the $\epsilon'$-first-order stationarity gives
    \begin{align*}
        \max_{p \in \mathcal{G}} T^f_p (\psi_k(\vtau^*)) - \sum_{p \in \mathcal{G}} \zeta^{+,*}_{k, p} T^f_p(\psi_k(\vtau^*)) \leq \epsilon'.
    \end{align*}
    Since $f = c_k t^3$, from \Cref{lem:grad_approx_err}, we have
    \begin{align*}
        \left|\sum_p \zeta^{+,*}_{k, p} f'(p) - f'(\psi_k(\vtau^*))\right| \leq \sqrt{18c_k(c_k \delta'^2 + \epsilon')}.
    \end{align*}
    On the other hand, the chain rule yields
    \begin{align*}
        \norm{\nabla_{\vtau} G_{k, \psi_k}^+(\vtau^{*}, \vzeta_k^{+,*}) - \nabla_{\vtau} \left(c_k \psi_k(\vtau^*)^3\right)}_\infty &\leq \left|\sum_p \zeta^{+,*}_{k, p} f'(p) - f'(\psi_k(\vtau^*))\right| \cdot \norm{\nabla_{\vtau} \psi_k(\vtau^*)}_{\infty} \\
        & \leq \sqrt{18c_k(c_k \delta'^2 + \epsilon')},
    \end{align*}
    where the last step holds because of the choice of the value of $\phi(\cdot)$ in \eqref{eq:phi_choice}, which implies 
    \[
    \norm{\nabla_{\vtau} \psi_k(\vtau^*)}_{\infty} = \norm{\nabla \phi_k(\tilde{x}_k, \tilde{y}_k, \tilde{z}_k)}_\infty \leq 1.
    \]
    The proof for $c_k < 0$ follows similarly.
\end{proof}

\begin{remark}
    Since each $\psi_k(\vtau)$ is an affine function of $\vtau,$ the approximation gadget $G_{k, \psi_k}^+(\vtau, \vzeta_k^+)$ and  $G_{k, \psi_k}^-(\vtau, \vzeta_k^-)$ are degree-2 multilinear functions. \Cref{lem:approx_gadget} effectively shows that, up to an error that scales polynomially in the size of the input, we can approximate the gradient of degree-3 terms $\psi(\vtau)^3$ through the approximation gadget $G_{k, \psi_k}^+(\vtau, \vzeta_k^+)$ and $G_{k, \psi_k}^-(\vtau, \vzeta_k^-)$.
\end{remark}

For any degree-3 term $\tilde{x}_k \tilde{y}_k \tilde{z}_k$ with coefficient $c_k$ in \eqref{eq:degree_3_obj}, we construct $7$ independent functions $\phi_{k, l}(\tilde{x}_k, \tilde{y}_k, \tilde{z}_k)$ where each $\phi_{k, l}$ takes one of the forms in \eqref{eq:phi_choice}. Moreover, the coefficients $c_{k, l}$ take the following values
\begin{equation*}
    \frac{27}{6} c_k, \quad -\frac{8}{6} c_k, \quad -\frac{8}{6} c_k, \quad -\frac{8}{6} c_k, \quad \frac{1}{6} c_k, \quad \frac{1}{6} c_k, \quad \frac{1}{6} c_k.
\end{equation*}
Then each degree-3 term $\tilde{x}_k \tilde{y}_k \tilde{z}$ can be written as 
\begin{equation*}
    c_k \tilde{x}_k \tilde{y}_k \tilde{z}_k = \sum_{l = 1}^ 7 c_{k, l} \phi_{k, l}^3(\tilde{x}_k, \tilde{y}_k, \tilde{z}_k) = \sum_{l = 1}^7 c_{k, l} \psi_{k, l}^3(\vtau).
\end{equation*}
Based on the sign of the coefficients $c_{k,l}$, we introduce approximation gadgets $G_{k, l, \psi_{k, l}}^+(\vtau, \vzeta^+_{k,l})$ if $c_{k,l} > 0$ and $G_{k, l, \psi_{k, l}}^-(\vtau, \vzeta^-_{k,l})$ if $c_{k,l} < 0.$ 
We define the final approximation to be
\begin{align}\label{eq:term_approx}
    G_{k}(\vtau ,\vzeta^+_k, \vzeta^-_k) \defeq \sum_{l: c_{k, l} > 0}G_{k, l, \psi_{k, l}}^+(\vtau, \vzeta^+_{k,l}) + \sum_{l: c_{k, l} < 0}G_{k, l, \psi_{k, l}}^-(\vtau, \vzeta^-_{k,l}).
\end{align}
We have the following lemma.
\begin{lemma}\label{lem:approx_grad_diff}
    Let the objective function be the approximation gadget $G_{k}$ for the term $c_k \tilde{x}_k \tilde{y}_k \tilde{z}_k$. For any $\epsilon'$-first-order stationary point, it holds that
    \begin{align*}
    \norm{\nabla_{\vtau} G_k(\vtau^*, \vzeta^{+*}_k, \vzeta^{-*}_k) - \nabla_{\vtau}(c_k \tilde{x}_k^* \tilde{y}_k^* \tilde{z}_k^*)}_{\infty} \leq 70 \sqrt{|c_k| (5|c_k| \delta'^2 + \epsilon')}.
    \end{align*}
\end{lemma}
\begin{proof}
    The proof follows by observing that $G_k$ is a sum of $7$ terms as in \eqref{eq:term_approx}. Combining \Cref{lem:approx_gadget} with the fact that $|c_{k,l}| \leq 5|c_k|$ yields the desired upper bound.
\end{proof}

Finally, let $F(\vtau)$ denote the original objective function in \eqref{eq:final_game}. Since $F(\vtau)$ is a degree-$3$ multilinear function, we can write it as
\begin{align*}
    F(\vtau) = F_{\leq 2}(\vtau) + \sum_k c_k \tilde{x}_k \tilde{y}_k \tilde{z}_k,
\end{align*}
where $F_{\leq 2}(\vtau)$ denotes the sum of all terms in $F(\vtau)$ with degree less than or equal to $2$. We have at most $|V| (4 m n + 2mn^2) \leq 3|V|mn^2$ different degree-3 terms in \eqref{eq:degree_3_obj}. We define the approximation function of $F$ using approximation gadget $G_k$ defined in \eqref{eq:term_approx} as
\begin{align*}
    \hat{F}(\vtau, \vzeta^+, \vzeta^-) \defeq F_{\leq 2}(\vtau) + \sum_{k} G_k (\vtau, \vzeta_k^+, \vzeta_k^-).
\end{align*}
Since $G_k$ is a degree-2 multilinear function, the approximation $\hat{F}(\vtau, \vzeta^+, \vzeta^-)$ is a degree-2 multilinear function. We proceed with the following lemma.
\begin{lemma} \label{lem:VI_solution_degree2}
    For every $\epsilon'$-first-order stationary point of the degree-2 multilinear approximation function $\hat{F}(\vtau, \vzeta^+, \vzeta^-)$, $\vtau^*$ is an $\left(\epsilon' + 600|V|^2m^2n^3K \sqrt{\delta'^2 + \epsilon'}\right)$-first-order stationary point of the original function \eqref{eq:final_game}.
\end{lemma}
\begin{proof}
    We have
    \begin{align*}
        \nabla_{\vtau} \hat{F}(\vtau^*, \vzeta^{+*}, \vzeta^{-*}) & = \nabla_{\vtau} F_{\leq 2}(\tau^*) + \sum_k \nabla_{\vtau} G_k(\vtau^*, \vzeta^{+*}_k, \vzeta^{-*}_k).
    \end{align*}
    From \Cref{lem:approx_grad_diff}, we have 
    \begin{align*}
        \norm{\nabla_{\vtau^*} \hat{F}(\vtau^*, \vzeta^{+*}, \vzeta^{-*}) - \nabla_{\vtau} F(\vtau^*)}_\infty &\leq \sum_{k} 70\sqrt{|c_k| (5|c_k|\delta'^2 + \epsilon')} \\
        &\leq 600 |V|mn^2K \sqrt{\delta'^2 + \epsilon'},
     \end{align*}
     where we use the fact that $k \leq 3|V|mn^2$ and $|c_k| \leq K.$

     Therefore, for any minimizing variable ${\vtau}_{\min}$ (i.e.  ${\vtau}_{\min} \in \left\{\vx, \vx', \vz\right\}$), 
     let $D_{\vtau_{\min}}$ denote the diameter of $\vtau_{\min},$ we have  $D_{\vtau_{\min}} \leq D_{\vx} \leq |V|mn.$ Therefore, let $\vtau_{\min}^*$ be an arbitrary minimizer variable in $\vtau^*$. From the $\epsilon'$-first-order stationarity condition, considering any deviation from $\vtau_{\min}^*$ to $\vtau_{\min}$, we have
     \begin{align*}
         &\left\langle \vtau_{\min} -\vtau^*_{\min}, \nabla_{\vtau_{\min}} \hat{F}(\vtau^*, \vzeta^{+*}, \vzeta^{-*})\right\rangle \geq -\epsilon'.
    \end{align*}
    This implies
    \begin{align*}
         \left\langle \vtau_{\min} -\vtau^*_{\min}, \nabla_{\vtau_{\min}} F(\vtau^*)\right\rangle \geq -\epsilon' - \norm{\vtau_{\min}^* - \vtau_{\min}}_1 \norm{\nabla_{\vtau} \hat{F}(\vtau^*, \vzeta^{+*}, \vzeta^{-*}) - \nabla_{\vtau} F(\vtau^*)}_\infty.
    \end{align*}
    Plugging in the bound, we have
    \begin{align*}
         &\left\langle \vtau_{\min} -\vtau^*_{\min}, \nabla_{\vtau_{\min}} F(\vtau^*)\right\rangle \geq -\epsilon' - 600|V|^2m^2n^3K \sqrt{\delta'^2 + \epsilon'}.
     \end{align*}
     The proof for maximizer variables $\left\{\vy_q, \vy'_q, \vz'_q, \valpha^{u,v,w}, \vbeta^{u, v, w}, \vgamma^{u,v,w}\right\}$ follows similarly.
\end{proof}

By tuning the value of $\epsilon'$ and $\delta'$, we have the following theorem.

\begin{theorem}\label{thm:degree_two_ppad}
    Let $\epsilon' \leq \epsilon^2/(10^{7} |V|^{4} m^{4} n^{6} K^{2})$ and $\delta' < \epsilon/(10^{4} |V|^2m^2n^3K)$ with $\epsilon < n^{-39}$, finding any $\epsilon'$-first-order stationary point $(\vtau^*, \vzeta^{+*}, \vzeta^{-*})$ of the degree-2 multilinear function $\hat{F}(\vtau, \vzeta^+, \vzeta^-)$ is \PPAD-complete.
\end{theorem}

\begin{proof}
    The $\PPAD$ membership follows again from~\citet{DSZ21}. To prove the hardness, by the choice of $\epsilon'$ and $\delta',$ \Cref{lem:VI_solution_degree2} implies that $\vtau^*$ is an $\epsilon$-first-order stationary point of the original objective defined in \eqref{eq:final_game}. The \PPAD-hardness then follows from \Cref{thm:degree_three_ppad}.
\end{proof}

\section{Reducing to a 2 vs.\ 2 team zero-sum game}
\label{sec:lawyer}

\Cref{thm:degree_two_ppad} shows that for some sufficiently small $\epsilon$ that is inversely polynomial in the size of the input, it is \PPAD-complete to find an $\epsilon$-first-order stationary point of the min-max optimization problem
\begin{align}\label{eq:degree_two_min_max}
    \min_{\vX \in \mathcal{X}} \max_{\vY \in \mathcal{Y}} \hat{F}(\vX, \vY),
\end{align}
where $\hat{F}$ is a degree-2 multilinear function. Moreover, $\mathcal{X}$ and $\mathcal{Y}$ are products of multiple simplex domains such that $\vX$ is the vector concatenation of all minimizers $(\vx, \vx', \vz, \vzeta^+)$ and $\vY$ is the vector concatenation of all maximizers $(\vy, \vy', \vz', \valpha, \vbeta, \vgamma, \vzeta^-).$  Some coordinates of the original variables have domain $[0, 1]$. For each of these coordinates, we represent it by a 2-dimensional vector $(t, 1 -t) \in \Delta_2$. In this section, we extend the \PPAD-hardness result to $2$ vs. $2$ team zero-sum games.

We construct a team zero-sum game from \eqref{eq:degree_two_min_max}. Specifically, we introduce two minimizers $\vX, \vX' \in \mathcal{X}$ and two maximizers $\vY, \vY' \in \mathcal{Y}$ and define the objective function $\hat{F}'(\vX, \vX', \vY, \vY')$ as follows.
\begin{itemize}
    \item For every degree-$2$ term $X_i X_j$ in the objective $\hat{F},$ we include $\frac{1}{2} (X_i X'_j + X_j X'_i)$ in $\hat{F}';$
    \item for every degree-$2$ term $Y_i Y_j$ in the objective $\hat{F},$ we include $\frac{1}{2} (Y_i Y'_j + Y_j Y'_i)$ in $\hat{F}';$
    \item for every degree-$2$ term $X_i Y_j$ in the objective $\hat{F},$ we include $\frac{1}{4} (X_i Y_j + X_i Y'_j + X_i'Y_j + X_i'Y_j')$ in $\hat{F}';$
    \item for every degree-$1$ term $X_i$ in the objective $\hat{F},$ we include $\frac{1}{2} (X_i + X'_i)$ in $\hat{F}';$
    \item for every degree-$1$ term $Y_i$ in the objective $\hat{F},$ we include $\frac{1}{2} (Y_i + Y'_i)$ in $\hat{F}'.$
\end{itemize}
By the construction of $\hat{F}',$ if $\vX = \vX'$ and $\vY = \vY',$ then it holds that $\hat{F}' (\vX,\vX, \vY, \vY) = \hat{F}(\vX, \vY)$ and $\nabla_{\vX} \hat{F}'(\vX, \vX, \vY, \vY) = \nabla_{\vX'} \hat{F}'(\vX, \vX, \vY, \vY) = \frac{1}{2} \nabla_{\vX} \hat{F}(\vX, \vY).$

However, there are two challenges for us to extend the \PPAD-hardness result to 2 vs.\ 2 team zero-sum games. The first challenge is that we need to design a gadget that enforces $\vX \simeq \vX'$ and $\vY \simeq \vY'$ at equilibrium. The second challenge is that in order to establish hardness in 2 vs.\ 2 team zero-sum games, we need to encode $\vX$ and $\vY$ to a single simplex domain. Since the domain $\mathcal{X}$ and $\mathcal{Y}$ are products of multiple simplices, we need to introduce a gadget that ensures $\vX$ and $\vY$ assign nearly equal probability mass to each of the original simplex blocks.

Below we introduce a gadget that addresses both challenges. Let $L_{\vX}$ and $L_{\vY}$ be the number of simplices concatenated in $\vX \in \mathcal{X}$ and $\vY \in \mathcal{Y}$ respectively.  Moreover, let $N_{\vX}$ and $N_{\vY}$ denote the number of coordinates (i.e., number of pure actions) in $\vX$ and $\vY$ respectively. Since $\vY$ contains more players and more strategies than $\vX$, we have $L_{\vY} \geq L_{\vX}$ and $N_{\vY} \geq N_{\vX}$. We introduce two minimizers $\tilde{\vX} \in \Delta_{2N_{\vX}N_{\vY}},$ $\tilde{\vX}' \in \Delta_{N_{\vX} L_{\vY}}$ and two maximizers
$\tilde{\vY} \in \Delta_{2N_{\vX}N_{\vY}},$ $\tilde{\vY}' \in \Delta_{N_{\vY}L_{\vX}}.$ One can think of each pure action in $\tilde{\vX}$ of the form $(i, j),$ where $i \in [N_{\vX}]$ and $j \in \{k^+, k^-\}$ for $k \in [N_{\vY}].$ Similarly, each pure action in $\tilde{\vY}$ takes the form $(i, j)$, where $i \in [N_{\vY}]$ and $j \in \{k^+, k^-\} $ for $ k \in [N_{\vX}].$ Intuitively, for each $i \in [N_{\vX}],$ one can think of player $\tilde{\vX}$ assigning $\sum_{k \in N_{\vY}} \left(\tilde{X}_{i, k^+} + \tilde{X}_{i, k^-}\right)$ probability mass to action $i$.

Each pure action of $\tilde{\vX}'$ takes the form $(i, j),$ where $i \in [N_{\vX}]$ and $j \in [L_{\vY}].$ For $i \in [N_{\vX}],$ the probability mass $\tilde{\vX}'$ assigns to the $i^{th}$ action is $\sum_{j \in [L_{\vY}]} \tilde{X}'_{i, j}$. Similarly, each pure action in $\tilde{\vY}'$ takes the form $(i, j)$ where $i \in [N_{\vY}]$ and $j \in [L_{\vX}],$ and $\tilde{\vY}'$ assigns $\sum_{j \in [L_\vX]} \tilde{Y}'_{i, j}$ probability mass to action $i.$

We define functions 
\[
\vX(\tilde{\vX})_i = L_{\vX}\left(\sum_{k \in [N_{\vY}]} \left(\tilde{X}_{i, k^+} + \tilde{X}_{i, k^-}\right)\right)
\]
and 
\[
\vY(\tilde{\vY})_i = L_{\vY} \left( \sum_{k \in [N_{\vX}]} \left(\tilde{Y}_{i, k^+} + \tilde{Y}_{i, k^-}\right)\right).
\]
Additionally, we let $\vX(\tilde{\vX})$ be the concatenation of all $\vX(\tilde{\vX})_i$ and similarly for $\vY(\tilde{\vY}).$  Therefore $\vX(\tilde{\vX})$ and $\vY(\tilde{\vY})$ map $\tilde{\vX}$ and $\tilde{\vY}$ to the same dimensions induced by domains $\mathcal{X}$ and $\mathcal{Y}.$ Similarly, we introduce $\vX'(\tilde{\vX}')_i = L_{\vX}\left(\sum_{k \in [L_{\vY}]} \tilde{X}'_{i, k}\right)$ and $\vY'(\tilde{\vY}')_i = L_{\vY}\left(\sum_{k \in [L_{\vX}]} \tilde{Y}'_{i, k}\right)$. $\vX'(\tilde{\vX}')$ and $\vY'(\tilde{\vY}')$ are defined accordingly.

For any $i \in [L_\vX]$, we also let $\mu_i^{\vX}(\tilde{\vX})$ denote the probability mass that player $\tilde{\vX}$ assigns to the $i^{th}$ simplex block. Specifically, 
\begin{align*}
    \mu_i^{\vX}(\tilde{\vX}) = \sum_{j: \text {action $j$ is in $i^{th}$ simplex block}  } \sum_{k \in [N_{\vY}]} \left(\tilde{X}_{j, k^+} + \tilde{X}_{j, k^-}\right) .
\end{align*}
Similarly, we define
\begin{align*}
    \mu_i^{\vY}(\tilde{\vY}) = \sum_{j: \text {action $j$ is in $i^{th}$ simplex block}  } \sum_{k \in [N_{\vX}]} \left(\tilde{Y}_{j, k^+} + \tilde{Y}_{j, k^-}\right).
\end{align*}

To ensure that $\vX(\tilde{\vX})$ and $\vX'(\tilde{\vX}')$ assign similar probability mass to each action $i \in [N_{\vX}]$ (similarly for $\vY(\tilde{\vY})$ and $\vY'(\tilde{\vY}')$), we introduce the following objective
\begin{align*}
    \Fcopy(\tilde{\vX}, \tilde{\vX}', \tilde{\vY}, \tilde{\vY}')
    = & - \sum_{i \in [N_{\vY}]} \left(\sum_{j \in [N_{\vX}]} \tilde{X}_{j, i^+} \left(\vY'(\tilde{\vY}')_i - \vY(\tilde{\vY})_i\right) + 
    \sum_{j \in [N_{\vX}]} \tilde{X}_{j, i^-} \left(\vY(\tilde{\vY})_i  - \vY'(\tilde{\vY}')_i\right)\right) \\
    & + \sum_{i \in [N_{\vX}]} \left(\sum_{j \in [N_{\vY}]} \tilde{Y}_{j, i^+} \left(\vX'(\tilde{\vX}')_i - \vX(\tilde{\vX})_i \right) + 
    \sum_{j \in [N_{\vY}]} \tilde{Y}_{j, i^-} \left(\vX(\tilde{\vX})_i  - \vX'(\tilde{\vX}')_i\right)\right).
\end{align*}
At a high level, we can think of $\tilde{X}_{j, i^+}, \tilde{X}_{j, i^-}$ and $\tilde{Y}_{j, i^+}, \tilde{Y}_{j, i^-}$ as playing roles similar to $\vz_+, \vz_-, \vz'_+, \vz'_-$ in our design of the $\Copy$ gadget in \Cref{sec:degree3}.

Furthermore, to ensure that $\vX(\tilde{\vX})$ and $\vY(\tilde{\vY})$ assign roughly equal probability mass to each simplex block, we design the following gadget following the same spirit as the generalized matching-pennies game~\citep{Daskalakis09:The}:
\begin{align*}
    \Fblock(\tilde{\vX}, \tilde{\vX}', \tilde{\vY}, \tilde{\vY}') = \sum_{i \in L_{\vX}} \mu_i^{\vX}(\tilde{\vX}) \cdot \left(\sum_{j \in N_{\vY}} \tilde{\vY}'_{j, i}\right) - \sum_{i \in L_{\vY}} \mu_{i}^{\vY}(\tilde{\vY}) \cdot \left(\sum_{j \in N_{\vX}}\tilde{\vX}'_{j, i} \right).
\end{align*}
Effectively, one can think of $(\tilde{\vX}, \tilde{\vY}')$ playing a generalized matching-pennies game and $(\tilde{\vY}, \tilde{\vX}')$ playing another generalized matching-pennies game.
We define the objective function to be
\begin{align*}
    \tilde{F}(\tilde{\vX}, \tilde{\vX}',\tilde{\vY},\tilde{\vY}') \defeq  &\hat{F}'(\vX(\tilde{\vX}), \vX'(\tilde{\vX}'), \vY(\tilde{\vY}), \vY'(\tilde{\vY}')) \\
    & + K_1 \Fcopy(\tilde{\vX}, \tilde{\vX}', \tilde{\vY}, \tilde{\vY}') + K_2\Fblock(\tilde{\vX}, \tilde{\vX}', \tilde{\vY}, \tilde{\vY}').
\end{align*}
Since $\vX(\cdot), \vX'(\cdot), \vY(\cdot), \vY'(\cdot)$ are all linear functions and $\hat{F}',\Fcopy, \Fblock$ are degree-2 multilinear functions, $\tilde{F}$ is a degree-2 multilinear function. 

\begin{remark}
    Since each of $\vX(\tilde{X}), \vX'(\tilde{X'}), \vY(\tilde{Y}),$ and $\vY'(\tilde{Y}')$ may not lie in a product of simplices, we slightly modify the domain of $\hat{F}'.$ Specifically, we define its domain to be the Euclidean space with the same dimension as $\mathcal{X} \times \mathcal{X} \times \mathcal{Y} \times \mathcal{Y}.$ Nevertheless, later in the proof of \Cref{thm:team_games_ppad}, we will construct $\bar{\vX} \in \mathcal{X}$ and $\bar{\vY} \in \mathcal{Y}$ to ensure the reduction goes through.
\end{remark}

We continue with the following lemma.
\begin{lemma} \label{lem:clone}
    Let $L_{\hat{F}'}$ be the Lipschitz constant for $\hat{F}'.$ Then $L_{\hat{F}'}$ scales polynomially in the size of the input. Let the objective function be $\tilde{F}(\tilde{\vX}, \tilde{\vX}',\tilde{\vY},\tilde{\vY}')$ with $K_1 > 2L_{\hat{F}'}N_{\vY}$. At any $\epsilon$-approximate Nash equilibrium, we have 
    \begin{align*}
        \norm{ \vX(\tilde{\vX}^*) - \vX'(\tilde{\vX}'^*)}_\infty \leq \frac{4\epsilon}{K_1}, \quad \text{and} \quad\norm{ \vY(\tilde{\vY}^*) - \vY'(\tilde{\vY}'^*)}_\infty \leq \frac{4\epsilon}{K_1}.
    \end{align*}
\end{lemma}
\begin{proof}
    Let $i \in [N_{\vX}]$ such that $|\vX(\tilde{\vX}^*)_i - \vX'(\tilde{\vX}'^*)_i|$ is maximized, we assume $\vX(\tilde{\vX}^*)_i - \vX'(\tilde{\vX}'^*)_i < 0,$ the proof for the other case follows symmetrically. Consider the deviation of $\tilde{\vY}^*$ to $\tilde{\vY}$ such that $\tilde{\vY}_{j, i^+} = \sum_{k \in N_{\vX}} \tilde{\vY}^*_{j, k^+} + \tilde{\vY}^*_{j, k^-}$ and $\tilde{\vY}_{j, i^-} = \tilde{\vY}_{j,k} = 0$ for $k \neq i.$ We have $\vY(\tilde{\vY}^*) = \vY(\tilde{\vY}).$ Therefore, it holds that 
    \begin{align*}
        \hat{F}'(\vX(\tilde{\vX}^*), \vX'(\tilde{\vX}'^*), \vY(\tilde{\vY}^*), \vY'(\tilde{\vY}'^*)) = \hat{F}'(\vX(\tilde{\vX}^*), \vX'(\tilde{\vX}'^*), \vY(\tilde{\vY}), \vY'(\tilde{\vY}'^*)),
    \end{align*} 
    and 
    \begin{align*}
        \Fblock(\tilde{\vX}^*, \tilde{\vX}'^*, \tilde{\vY}^*, \tilde{\vY}'^*) = \Fblock(\tilde{\vX}^*, \tilde{\vX}'^*, \tilde{\vY}, \tilde{\vY}'^*).
    \end{align*} Moreover, the $\epsilon$-approximate Nash equilibrium condition of $\tilde{\vY}^*$ gives 
    \begin{align} \label{eq:Y_deviation}
    & \sum_{i \in [N_{\vX}]} \left(\sum_{j \in [N_{\vY}]} \tilde{Y}^*_{j, i^+} \left(\vX'(\tilde{\vX}'^*)_i - \vX(\tilde{\vX}^*)_i \right) + 
    \sum_{j \in [N_{\vY}]} \tilde{Y}^*_{j, i^-} \left(\vX(\tilde{\vX}^*)_i  - \vX'(\tilde{\vX}'^*)_i\right)\right) \\
    &\geq \norm{\vX(\tilde{\vX}^*) - \vX'(\tilde{\vX}'^*)}_{\infty} - \frac{\epsilon}{K_1}.
    \end{align}
    On the other hand, consider the deviation of $\tilde{\vX}'^*$ to some $\tilde{\vX}'$ such that $\vX'(\tilde{\vX}') = \vX(\tilde{\vX}^*)$ and $\sum_{j \in N_{\vX}} \tilde{\vX}'^*_{j, i} = \sum_{j \in N_{\vX}} \tilde{\vX}'_{j, i}$ for all $i \in L_{\vY}$. It then holds that
    \begin{align*}
        \Fblock(\tilde{\vX}^*, \tilde{\vX}'^*, \tilde{\vY}^*, \tilde{\vY}'^*) = \Fblock(\tilde{\vX}^*, \tilde{\vX}', \tilde{\vY}^*, \tilde{\vY}'^*).
    \end{align*}
    Therefore, since $\tilde{\vX}'^*$ is part of an $\epsilon$-approximate Nash equilibrium, we have
    \begin{align*}
    & \hat{F}'(\vX(\tilde{\vX}^*), \vX'(\tilde{\vX}'^*), \vY(\tilde{\vY}^*), \vY'(\tilde{\vY}'^*)) + K_1 \Fcopy(\tilde{\vX}^*, \tilde{\vX}'^*, \tilde{\vY}^*, \tilde{\vY}'^*) \\
    &\leq \hat{F}'(\vX(\tilde{\vX}^*), \vX(\tilde{\vX}^*), \vY(\tilde{\vY}^*), \vY'(\tilde{\vY}'^*)) +K_1 \Fcopy(\tilde{\vX}^*, \tilde{\vX}', \tilde{\vY}^*, \tilde{\vY}'^*)+\epsilon.
    \end{align*}
    This gives
    \begin{align*}
       & K_1\left(\sum_{i \in [N_{\vX}]} \left(\sum_{j \in [N_{\vY}]} \tilde{Y}^*_{j, i^+} \left(\vX'(\tilde{\vX}'^*)_i - \vX(\tilde{\vX}^*)_i \right) + 
    \sum_{j \in [N_{\vY}]} \tilde{Y}^*_{j, i^-} \left(\vX(\tilde{\vX}^*)_i  - \vX'(\tilde{\vX}'^*)_i\right)\right)\right) \\
    & + \hat{F}'(\vX(\tilde{\vX}^*), \vX'(\tilde{\vX}'^*), \vY(\tilde{\vY}^*), \vY'(\tilde{\vY}'^*)) - \hat{F}'(\vX(\tilde{\vX}^*), \vX(\tilde{\vX}^*), \vY(\tilde{\vY}^*), \vY'(\tilde{\vY}'^*)) \leq \epsilon.
    \end{align*}
    Taking the fact that $\hat{F}'$ is $L_{\hat{F}'}$-Lipschitz continuous for each coordinate of $\tilde{\vX}',$ we have
    \begin{align*}
        & K_1\left(\sum_{i \in [N_{\vX}]} \left(\sum_{j \in [N_{\vY}]} \tilde{Y}^*_{j, i^+} \left(\vX'(\tilde{\vX}'^*)_i - \vX(\tilde{\vX}^*)_i \right) + 
    \sum_{j \in [N_{\vY}]} \tilde{Y}^*_{j, i^-} \left(\vX(\tilde{\vX}^*)_i  - \vX'(\tilde{\vX}'^*)_i\right)\right)\right) \\
    & - L_{\hat{F}'}N_{\vX}\norm{\vX(\tilde{\vX}^*) - \vX'(\tilde{\vX}'^*)}_{\infty} \leq \epsilon.
    \end{align*}
    Combining with \eqref{eq:Y_deviation}, it holds that
    \begin{align*}
        (K_1 - L_{\hat{F}'}N_{\vX}) \norm{\vX(\tilde{\vX}^*) - \vX'(\tilde{\vX}'^*)}_{\infty} \leq 2 \epsilon.
    \end{align*}
    By setting $K_1 > 2L_{\hat{F}'} N_{\vY} \geq 2L_{\hat{F}'} N_{\vX},$ we have
    \begin{align*}
        \norm{\vX(\tilde{\vX}^*) - \vX'(\tilde{\vX}'^*)}_{\infty} \leq \frac{4\epsilon}{K_1}.
    \end{align*}
    A similar proof works for $\tilde{\vY}^*$ and $\tilde{\vY}'^*.$
\end{proof}

It remains to prove that $\vX(\tilde{\vX}^*)$ and $\vY(\tilde{\vY}^*)$ assign roughly equal probability mass on each simplex block. We prove this in the following lemma.
\begin{lemma} \label{lem:block_mass}
    Let the objective function be $\tilde{F}(\tilde{\vX}, \tilde{\vX}',\tilde{\vY},\tilde{\vY}')$. Let $K_2 \geq \frac{L_{\vY}(N_{\vY} L_{\hat{F'}} L_{\vY} + L_{\vY}K_1 + 2\epsilon)}{\epsilon}$ which scales polynomially in the size of the input. For any $\epsilon$-approximate Nash equilibrium, it holds that
    \begin{align*}
        \max_{i \in L_{\vX}} \left|\mu_i^{\vX}(\tilde{\vX}^*) - \frac{1}{L_{\vX}} \right| \leq \epsilon, \quad \text{and} \quad \max_{i \in L_{\vY}} \left|\mu_i^{\vY}(\tilde{\vY}^*) - \frac{1}{L_{\vY}} \right| \leq \epsilon
    \end{align*}
\end{lemma}
\begin{proof}
    Let $i \in [L_{\vX}]$ be the simplex block where $\mu_i^{\vX}(\tilde{\vX}^*)$ is maximized. We consider the deviation for $\tilde{\vY}'^*$ to $\tilde{\vY}'$ such that
    \begin{align*}
        \tilde{Y}'_{j, i} = \sum_{k \in [L_{\vX}]} \tilde{Y}'^*_{j, k} \quad \text{and} \quad \tilde{Y}'_{j, k} = 0 \quad \forall k \neq i. 
    \end{align*}
    Since $\vY'(\tilde{\vY}'^*) = \vY'(\tilde{\vY}')$, we have
    \begin{align*}
        \hat{F}'(\vX(\tilde{\vX}^*), \vX'(\tilde{\vX}'^*), \vY(\tilde{\vY}^*), \vY'(\tilde{\vY}'^*)) = \hat{F}'(\vX(\tilde{\vX}^*), \vX'(\tilde{\vX}'^*), \vY(\tilde{\vY}^*), \vY'(\tilde{\vY}'))
    \end{align*}
    and
    \begin{align*}
    \Fcopy(\tilde{\vX}^*, \tilde{\vX}'^*, \tilde{\vY}^*, \tilde{\vY}'^*) = \Fcopy(\tilde{\vX}^*, \tilde{\vX}'^*, \tilde{\vY}^*, \tilde{\vY}').
    \end{align*}
    Since $\tilde{\vY}'^*$ is part of an $\epsilon$-approximate Nash equilibrium, we have
    \begin{align*}
        \Fblock(\tilde{\vX}^*, \tilde{\vX}'^*, \tilde{\vY}^*, \tilde{\vY}') \leq \Fblock(\tilde{\vX}^*, \tilde{\vX}'^*, \tilde{\vY}^*, \tilde{\vY}'^*) +\frac{\epsilon}{K_2}.
    \end{align*}
    Plugging in the value of $\tilde{\vY}',$ we have
    \begin{align}\label{eq:bound_on_simplex_block}
        \max_{i \in  [L_{\vX}]}\mu_{i}^{\vX}(\tilde{\vX}^*) \leq \sum_{i \in L_{\vX}} \mu_i^{\vX}(\tilde{\vX}^*) \cdot \left(\sum_{j \in N_{\vY}} \tilde{\vY}'^*_{j, i}\right) +\frac{\epsilon}{K_2}.
    \end{align}
    On the other hand, let $j \in [L_{\vX}]$ be the simplex block where $\sum_{k \in N_{\vY}} \tilde{Y}'^*_{k, j}$ is minimized. Since $\tilde{\vY}'^* \in \Delta_{N_{\vY} L_{\vX}},$ we have
    \begin{align*}
        \min_{j} \sum_{k \in N_{\vY}} \tilde{Y}'^*_{k, j} \leq \frac{1}{L_{\vX}}.
    \end{align*}
    Choose an action $i$ in simplex $j$, and consider the deviation of the minimizer $\tilde{\vX}^*$ to $\tilde{\vX}$ such that
    \begin{align*}
        \tilde{X}_{i, k^+} = \sum_{l \in N_{\vX}}  \tilde{X}^*_{l, k^+}, \quad \tilde{X}_{i, k^-} = \sum_{l \in N_{\vX}}  \tilde{X}^*_{l, k^-},
    \end{align*}
    and 
    \begin{align*}
        \tilde{X}_{l, k^+} = \tilde{X}_{l, k^-} = 0 \quad \forall l \neq i.
    \end{align*}
    From the definition of $\epsilon$-approximate Nash equilibrium, we have
    \begin{align*}
    & \sum_{i \in L_{\vX}} \mu_i^{\vX}(\tilde{\vX}^*) \cdot \left(\sum_{j \in N_{\vY}} \tilde{\vY}'^*_{j, i}\right)  \\
    & \leq \min_{j} \sum_{k \in N_{\vY}} \tilde{Y}'^*_{k, j} \\
    & \quad + \frac{1}{K_2} \left(\hat{F}'(\vX(\tilde{\vX}^*), \vX'(\tilde{\vX}'^*), \vY(\tilde{\vY}^*), \vY'(\tilde{\vY}'^*)) - \hat{F}'(\vX(\tilde{\vX}), \vX'(\tilde{\vX}'^*), \vY(\tilde{\vY}^*), \vY'(\tilde{\vY}'^*))\right) \\
    & \quad + \frac{K_1}{K_2} \left(\Fcopy(\tilde{\vX}^*, \tilde{\vX}'^*, \tilde{\vY}^*, \tilde{\vY}'^*) - \Fcopy(\tilde{\vX}, \tilde{\vX}'^*, \tilde{\vY}^*, \tilde{\vY}'^*)\right) + \frac{\epsilon}{K_2} \\
    & \leq \frac{1}{L_{\vX}} + \frac{N_{\vX}L_{\hat{F}'} L_{\vX}}{K_2} + \frac{L_{\vX}K_1}{K_2} + \frac{\epsilon}{K_2}.
\end{align*}
    Combining with \eqref{eq:bound_on_simplex_block}, we get
    \begin{align*}
        \max_{i \in  [L_{\vX}]}\mu_{i}^{\vX}(\tilde{\vX}^*) \leq \frac{1}{L_{\vX}} + \frac{N_{\vX} L_{\hat{F}'}L_{\vX} + L_{\vX}K_1 + 2\epsilon}{K_2}.
    \end{align*}
    Since there are $L_{\vX}$ simplex blocks in $\tilde{\vX}^*, $ we conclude that
    \begin{align*}
        \max_{i \in [L_{\vX}]} \left|\mu_{i}^{\vX}(\tilde{\vX}^*) - \frac{1}{L_{\vX}}\right| \leq \frac{L_{\vX}(N_{\vX} L_{\hat{F'}} L_{\vX} + L_{\vX}K_1 + 2\epsilon)}{K_2}.
    \end{align*}
    Setting $K_2 \geq \frac{L_{\vY}(N_{\vY} L_{\hat{F'}}L_{\vY} + K_1 L_{\vY} + 2\epsilon)}{\epsilon} \geq \frac{L_{\vX}(N_{\vX} L_{\hat{F'}} L_{\vX} + K_1 L_{\vX} + 2\epsilon)}{\epsilon}$ gives the desired bound. The proof for $\tilde{\vY}^*$ follows similarly.
\end{proof}

Finally, we state our main theorem for 2 vs.\ 2 team zero-sum games. 
\begin{theorem} \label{thm:team_games_ppad}
    For some $\epsilon$ that is inversely polynomial in the size of the input, finding an $\epsilon$-approximate Nash equilibrium in 2 vs.\ 2 polymatrix team zero-sum games is \PPAD-complete.
\end{theorem}
\begin{proof}
    The \PPAD-membership follows from \citet{Daskalakis09:The}. To prove the \PPAD-hardness, we consider the utility for maximizers $\tilde{\vY}$ and $\tilde{\vY}'$ to be $\tilde{F}(\tilde{\vX}, \tilde{\vX}',\tilde{\vY},\tilde{\vY}')$, whereas minimizers $\tilde{\vX}, \tilde{\vX}'$ aim to minimize this objective. Consider an $\epsilon''$-approximate Nash equilibrium of such a game. From \Cref{lem:clone}, we have 
    \begin{align}\label{eq:clone}
         \norm{ \vX(\tilde{\vX}^*) - \vX'(\tilde{\vX}'^*)}_\infty \leq \frac{4\epsilon''}{K_1}, \quad \text{and} \quad\norm{ \vY(\tilde{\vY}^*) - \vY'(\tilde{\vY}'^*)}_\infty \leq \frac{4\epsilon''}{K_1}.
    \end{align}
    Moreover, from \Cref{lem:block_mass}, for all simplex block $i \in [L_{\vX}]$ and $j \in [L_{\vY}],$ it holds that
    \begin{align} \label{eq:block_mass}
        \left|\mu_i^{\vX}(\tilde{\vX}^*) - \frac{1}{L_{\vX}}\right| \leq \epsilon'' \quad \text{and} \quad 
        \left|\mu_j^{\vY}(\tilde{\vY}^*) - \frac{1}{L_{\vY}}\right| \leq \epsilon''. 
    \end{align}
    From $\tilde{\vX}^*$ and $\tilde{\vY}^*,$ we construct two decoding points $\bar{\vX} \in \mathcal{X}$ and $\bar{\vY} \in \mathcal{Y}$. Specifically, for action $i \in [N_{\vX}]$ that belongs to simplex $j \in [L_{\vX}],$ we have
    \begin{align*}
        \bar{X}_i^j = \frac{\sum_{k \in N_{\vY}} \left(\tilde{X}^*_{i, k^+} + \tilde{X}^*_{i, k^-}\right)}{\mu_{j}^{\vX}(\tilde{\vX}^*)}.
    \end{align*}
    By setting $\epsilon'' \leq \frac{1}{2L_{\vX}},$ it holds that $\bar{X}_i^j \geq 0$ and $\sum_{i} \bar{X}_i^j = 1$. Define $\bar{\vX}^j$ to be the vector concatenation of $\bar{X}^j_i,$ and $\bar{\vX}$ to be the vector concatenation of $\bar{\vX}^j$. We have $\bar{\vX}^j \in \Delta$ and $\bar{\vX} \in \mathcal{X}$. Similarly, we construct
    \begin{align*}
        \bar{Y}_i^j = \frac{\sum_{k \in N_{\vX}} \left(\tilde{Y}_{i, k^+} + \tilde{Y}_{i, k^-}\right)}{\mu_{j}^{\vY}(\tilde{\vY^*})}.
    \end{align*}
    From
    \eqref{eq:block_mass}, we have
    \begin{align}\label{eq:clone_3}
        \norm{ \vX(\tilde{\vX}^*) - \bar{\vX}}_{\infty} \leq L_{\vX}\epsilon'', \quad \text{and} \quad
         \norm{\vY(\tilde{\vY}^*) - \bar{\vY}}_{\infty} \leq L_{\vY}\epsilon''.
    \end{align}
    Moreover, combining with \eqref{eq:clone}, it holds that
    \begin{align}\label{eq:clone_2}
        \norm{ \vX'(\tilde{\vX}'^*) - \bar{\vX}}_{\infty} \leq L_{\vX}\epsilon'' + \frac{4\epsilon''}{K_1}, \quad \text{and} \quad
         \norm{ \vY'(\tilde{\vY}'^*) - \bar{\vY}}_{\infty} \leq L_{\vY}\epsilon'' +\frac{4\epsilon''}{K_1}.
    \end{align}
    We proceed to show that $(\bar{\vX}, \bar{\vY})$ is an approximate first-order stationary point of \eqref{eq:degree_two_min_max}. For any deviation $\vX^{\dagger} \in \mathcal{X}$, for all $i \in [L_{\vX}]$ and for all action $j$ in simplex $i$, we construct $\tilde{\vX}^{\dagger} \in \Delta_{2N_{\vX} N_{\vY}}$ as follows:
    \begin{align*}
        \sum_{k \in [N_{\vY}]} \left(\tilde{X}^\dagger_{j,k^+} + \tilde{X}^\dagger_{j,k^-}\right) = \mu_{i}^{\vX}(\tilde{\vX}^*) X_j^\dagger.
    \end{align*} 
    From this construction, we have 
    \begin{align*}
        \mu_{i}^{\vX}(\tilde{\vX}^\dagger) = \mu_{i}^{\vX}(\tilde{\vX}^*) \quad \forall i \in [L_{\vX}].
    \end{align*}
    It also holds that
    \begin{align}\label{eq:decode_distance}
        \norm{\vX(\tilde{\vX}^\dagger) - \vX^\dagger}_{\infty} \leq \left|L_{\vX} \mu_{i}^{\vX}(\tilde{\vX}^*)- 1\right| \norm{\vX^{\dagger}}_{\infty} \leq L_{\vX}\epsilon'',
    \end{align}
    where the last step follows from \eqref{eq:block_mass}. Moreover, for each action $j$, we distribute the probability mass between $\tilde{X}^{\dagger}_{j, k^+}$ and $\tilde{X}^{\dagger}_{j, k^-}$
    such that
    \begin{align} \label{eq:copy_1}
        \sum_{j \in [N_{\vX}]} \tilde{X}^\dagger_{j, k^+} = \sum_{j \in [N_{\vX}]} \tilde{X}^*_{j, k^+}, \quad \text{and} \quad \sum_{j \in [N_{\vX}]} \tilde{X}^\dagger_{j, k^-} = \sum_{j \in [N_{\vX}]} \tilde{X}^*_{j, k^-}.
    \end{align}
    Therefore, it holds that
    \begin{align} \label{eq:block_1}
        \Fblock(\tilde{\vX}^\dagger, \tilde{\vX}'^*,\tilde{\vY}^*,
        \tilde{\vY}'^*) = \Fblock(\tilde{\vX}^*, \tilde{\vX}'^*,\tilde{\vY}^*,
        \tilde{\vY}'^*).
    \end{align}
    Additionally, we construct $\tilde{\vX}'^\dagger \in \Delta_{N_{\vX}L_{\vY}}$ such that for action $j$ in simplex $i$,
    \begin{align*}
        \tilde{X}'^\dagger_{j, k} = \mu_{i}^\vX(\tilde{\vX}^*)X_j^\dagger\left(\sum_{l \in [N_{\vX}]} \tilde{X}'^*_{l, k}\right).
    \end{align*}
    This ensures that $\vX(\tilde{\vX}^\dagger) = \vX'(\tilde{\vX}'^\dagger)$ and also
    \begin{align*}
        \sum_{i \in [N_{\vX}]} \tilde{X}'^\dagger_{i, j} = \sum_{i \in [N_{\vX}]} \tilde{X}'^*_{i, j} \quad \forall j \in [L_{\vY}].
    \end{align*}
    Therefore, we also have 
    \begin{align}\label{eq:block_2}
        \Fblock(\tilde{\vX}^*, \tilde{\vX}'^\dagger,\tilde{\vY}^*,
        \tilde{\vY}'^*) = \Fblock(\tilde{\vX}^*, \tilde{\vX}'^*,\tilde{\vY}^*,
        \tilde{\vY}'^*).
    \end{align}
    Since $\tilde{\vX}^*, \tilde{\vX}'^*$ is part of an $\epsilon''$-approximate Nash equilibrium, it holds that
    \begin{align}\label{eq:VI_deviation}
        \tilde{F}(\tilde{\vX}^\dagger, \tilde{\vX}'^*, \tilde{\vY}^*, \tilde{\vY}'^*) +  \tilde{F}(\tilde{\vX}^*, \tilde{\vX}'^\dagger, \tilde{\vY}^*, \tilde{\vY}'^*) - 2  \tilde{F}(\tilde{\vX}^*, \tilde{\vX}'^*, \tilde{\vY}^*, \tilde{\vY}'^*) \geq - 2\epsilon''.
    \end{align}
    From \eqref{eq:block_1} and \eqref{eq:block_2}, we know that $\Fblock$ parts are canceling out. For $\Fcopy$ parts, using the fact from \eqref{eq:copy_1}, we have
    \begin{align*}
        &\left|\Fcopy(\tilde{\vX}^\dagger , \tilde{\vX}'^*, \tilde{\vY}^*,
        \tilde{\vY}'^*) + \Fcopy(\tilde{\vX}^* , \tilde{\vX}'^\dagger, \tilde{\vY}^*,
        \tilde{\vY}'^*) - 2 \Fcopy(\tilde{\vX}^* , \tilde{\vX}'^*, \tilde{\vY}^*,
        \tilde{\vY}'^*) \right|\\ 
        =& \Bigg| \sum_{i \in [N_{\vX}]} \left(\sum_{j \in [N_{\vY}]} \tilde{Y}^*_{j, i^+} \left(\vX'(\tilde{\vX}'^*)_i - \vX(\tilde{\vX}^\dagger)_i \right) + 
    \sum_{j \in [N_{\vY}]} \tilde{Y}^*_{j, i^-} \left(\vX(\tilde{\vX}^\dagger)_i  - \vX'(\tilde{\vX}'^*)_i\right)\right)\\
    & +\sum_{i \in [N_{\vX}]} \left(\sum_{j \in [N_{\vY}]} \tilde{Y}^*_{j, i^+} \left(\vX'(\tilde{\vX}'^\dagger)_i - \vX(\tilde{\vX}^*)_i \right) + 
    \sum_{j \in [N_{\vY}]} \tilde{Y}^*_{j, i^-} \left(\vX(\tilde{\vX}^*)_i  - \vX'(\tilde{\vX}'^\dagger)_i\right)\right)\\
    & - 2\left(\sum_{i \in [N_{\vX}]} \left(\sum_{j \in [N_{\vY}]} \tilde{Y}^*_{j, i^+} \left(\vX'(\tilde{\vX}'^*)_i - \vX(\tilde{\vX}^*)_i \right) + 
    \sum_{j \in [N_{\vY}]} \tilde{Y}^*_{j, i^-} \left(\vX(\tilde{\vX}^*)_i  - \vX'(\tilde{\vX}'^*)_i\right)\right)\right) \Bigg|
    \end{align*}
    Since $\vX(\tilde{\vX}^\dagger) = \vX'(\tilde{\vX}'^\dagger),$ the equality above is 
    \begin{align*}
        &\left|\sum_{i \in [N_{\vX}]} \left(\sum_{j \in [N_{\vY}]} \tilde{Y}^*_{j, i^+} \left(\vX'(\tilde{\vX}'^*)_i - \vX(\tilde{\vX}^*)_i \right) + 
    \sum_{j \in [N_{\vY}]} \tilde{Y}^*_{j, i^-} \left(\vX(\tilde{\vX}^*)_i  - \vX'(\tilde{\vX}'^*)_i\right)\right)\right| \\
    & \leq \norm{\vX(\tilde{\vX}^*) - \vX'(\tilde{\vX'}^*)}_{\infty} \leq \frac{4\epsilon''}{K_1},
    \end{align*}
    where the last step follows from \eqref{eq:clone}. Therefore, from \eqref{eq:VI_deviation}, it holds that
    \begin{align*}
        & \hat{F}'(\vX(\tilde{\vX}^\dagger), \vX'(\tilde{\vX}'^*), \vY(\tilde{\vY}^*), \vY'(\tilde{\vY}'^*)) + \hat{F}'(\vX(\tilde{\vX}^*), \vX'(\tilde{\vX}'^\dagger), \vY(\tilde{\vY}^*), \vY'(\tilde{\vY}'^*))\\
        & - 2\hat{F}'(\vX(\tilde{\vX}^*), \vX(\tilde{\vX}'^*), \vY(\tilde{\vY}^*), \vY'(\tilde{\vY}'^*)) \\
        & \geq -2\epsilon'' - K_1\frac{4\epsilon''}{K_1} = -6\epsilon''.
    \end{align*}
    Since we have $\vX(\tilde{\vX}^\dagger) =\vX'(\tilde{\vX}'^\dagger),$ use  \eqref{eq:decode_distance}
    along with the fact that $\hat{F}' (\vX, \vX', \vY, \vY')$ is linear with respect to $\vX$ and $\vX'$ individually, we get
    \begin{align*}
        & \left\langle\vX^\dagger - \vX(\tilde{\vX}^*), \nabla_{\vX} \hat{F}'\left(\vX(\tilde{\vX}^* ), \vX'(\tilde{\vX}'^* ),\vY(\tilde{\vY}^* ),\vY'(\tilde{\vY}'^* )\right)\right\rangle \\
        & + \left\langle\vX^\dagger - \vX'(\tilde{\vX}'^*), \nabla_{\vX'} \hat{F}'\left(\vX(\tilde{\vX}^* ), \vX'(\tilde{\vX}'^* ),\vY(\tilde{\vY}^* ),\vY'(\tilde{\vY}'^* )\right)\right\rangle \geq -6\epsilon'' - 2L_{\hat{F}'} N_{\vX}L_{\vX}\epsilon''.
    \end{align*}
    Using \eqref{eq:clone} and the fact that $\hat{F}'$ is $L_{\hat{F}'}$-Lipschitz continuous, we have
    \begin{align*}
&\Big\langle
    \vX^\dagger - \vX(\tilde{\vX}^*),
    \nabla_{\vX} \hat{F}'\left(
        \vX(\tilde{\vX}^*),
        \vX'(\tilde{\vX}'^*),
        \vY(\tilde{\vY}^*),
        \vY'(\tilde{\vY}'^*)
    \right)
\\
&\qquad
    + \nabla_{\vX'} \hat{F}'\left(
        \vX(\tilde{\vX}^*),
        \vX'(\tilde{\vX}'^*),
        \vY(\tilde{\vY}^*),
        \vY'(\tilde{\vY}'^*)
    \right)
\Big\rangle
\\
&\geq
    -6\epsilon''
    - 2L_{\hat{F}'} N_{\vX}L_{\vX}\epsilon''
    - \frac{4L_{\hat{F}'}N_{\vX}\epsilon''}{K_1}
\\
&\geq
    -10 \epsilon''
    - 2L_{\hat{F}'} N_{\vX} L_{\vX} \epsilon'',
\end{align*}
    where the last step follows from our choice of $K_1$ such that $K_1 \geq L_{\hat{F}'} N_{\vX}.$ Now from our construction of $\hat{F}',$ let $\ell_{\hat{F}'}$ be the smoothness coefficient of $\hat{F}'$, which is polynomial in the size of the input. We have 
    \begin{align*}
        & \norm{\nabla_{\vX} \hat{F}'\left(\vX(\tilde{\vX}^* ), \vX'(\tilde{\vX}'^* ),\vY(\tilde{\vY}^* ),\vY'(\tilde{\vY}'^* )\right) -\nabla_{\vX} \hat{F}'\left(\bar{\vX}, \bar{\vX}, \bar{\vY}, \bar{\vY}\right)}_{\infty} \\
        & \leq \ell_{\hat{F}'} \left(\norm{\vX(\tilde{\vX}^* ) - \bar{\vX}}_{\infty} + \norm{\vX'(\tilde{\vX}'^* ) - \bar{\vX}}_{\infty} + \norm{\vY(\tilde{\vY}^* ) - \bar{\vY}}_{\infty} + \norm{\vY'(\tilde{\vY}'^* ) - \bar{\vY}}_{\infty}\right) \\
        & \leq \ell_{\hat{F}'} \left(\frac{16\epsilon''}{K_1} + 2(L_{\vX} + L_{\vY})\epsilon''\right),
    \end{align*}
    where the last step follows from \eqref{eq:clone} and \eqref{eq:clone_2}. Similarly,
    \begin{align*}
        \norm{\nabla_{\vX'} \hat{F}'\left(\vX(\tilde{\vX}^* ), \vX'(\tilde{\vX}'^* ),\vY(\tilde{\vY}^* ),\vY'(\tilde{\vY}'^* )\right) -\nabla_{\vX} \hat{F}'\left(\bar{\vX}, \bar{\vX}, \bar{\vY}, \bar{\vY}\right)}_{\infty} \\ \leq \ell_{\hat{F}'} \left(\frac{16\epsilon''}{K_1} + 2(L_{\vX} + L_{\vY})\epsilon''\right).
    \end{align*}
    Therefore,
    \begin{align*}
        \Big\langle\vX^\dagger - \vX(\tilde{\vX}^*), 2 \nabla_{\vX} \hat{F}'\left(\bar{\vX}, \bar{\vX}, \bar{\vY}, \bar{\vY}\right) \Big \rangle\geq -10\epsilon'' - 2L_{\hat{F}'} N_{\vX}L_{\vX}\epsilon''- N_{\vX}\ell_{\hat{F}'}\left(\frac{32\epsilon''}{K_1} + 4(L_{\vX} + L_{\vY})\epsilon''\right).
    \end{align*}
    Moreover, we have
    \begin{align*}
        \vX^\dagger - \vX(\tilde{\vX}^*) = \vX^\dagger - \bar{\vX} + \bar{\vX} - \vX(\tilde{\vX}^*).
    \end{align*}
    From \eqref{eq:clone_3}, we have
    \begin{align*}
        \Big\langle\vX^\dagger - \bar{\vX}, 2 \nabla_{\vX} \hat{F}'\left(\bar{\vX}, \bar{\vX}, \bar{\vY}, \bar{\vY}\right)\Big\rangle \geq -10\epsilon'' - 3L_{\hat{F}'} N_{\vX}L_{\vX}\epsilon'' - N_{\vX}\ell_{\hat{F}'}\left(\frac{32\epsilon''}{K_1} + 4(L_{\vX} + L_{\vY})\epsilon''\right).
    \end{align*}
    Finally, $2 \nabla_{\vX} \hat{F}'\left(\bar{\vX}, \bar{\vX}, \bar{\vY}, \bar{\vY}\right) = \nabla_{\vX} \hat{F}(\bar{\vX}, \bar{\vY})$. By setting $\epsilon''$ to be inversely polynomial to the size of the input such that $10\epsilon'' + 3L_{\hat{F}'} N_{\vX}L_{\vX}\epsilon'' + N_{\vX}\ell_{\hat{F}'}\left(\frac{32\epsilon''}{K_1} + 4(L_{\vX} + L_{\vY})\epsilon''\right)  \leq \epsilon'$, for $\epsilon'$ in \Cref{thm:degree_two_ppad}, we show that for all $\vX^\dagger \in \mathcal{X},$
    \begin{align*}
        \Big\langle\vX^\dagger - \bar{\vX},  \nabla_{\vX} \hat{F}\left(\bar{\vX}, \bar{\vY}\right)\Big\rangle \geq -\epsilon'.
    \end{align*}
    Similarly, one can show for all $\vY^{\dagger} \in \mathcal{Y},$
    \begin{align*}
        \Big\langle\vY^\dagger - \bar{\vY},  \nabla_{\vY} \hat{F}\left(\bar{\vX}, \bar{\vY}\right)\Big\rangle \leq \epsilon'.
    \end{align*}
    Therefore, from an $\epsilon''$-approximate Nash equilibrium of the 2 vs.\ 2 polymatrix team zero-sum game with objective function $\tilde{F}(\tilde{\vX}, \tilde{\vX}',\tilde{\vY},\tilde{\vY}'),$ one can recover an $\epsilon'$-first-order stationary point of the min-max optimization problem defined in \eqref{eq:degree_two_min_max}. The \PPAD-hardness then follows from \Cref{thm:degree_two_ppad}.
\end{proof}
\section{Conclusions and future research}

We established that computing approximate Nash equilibria in team zero-sum games is \PPAD-complete even in the polymatrix setting with a $2$ vs.\ $2$ structure. As a result, team zero-sum games are computationally as hard as general-sum games, despite maintaining a global adversarial structure. As a byproduct, we also showed that computing first-order stationary points in min-max optimization is \PPAD-complete even for quadratic objectives, strengthening the recent breakthrough result of~\citet{Bernasconi26:Complexity}.

Our results leave open an important question. Specifically, the hardness established here applies under an inverse polynomial precision, ruling out a \emph{fully} polynomial-time approximation scheme ($\FPTAS$)---unless $\P = \PPAD$. It remains unclear whether a $\PTAS$ exists. Can one adapt the hardness result of~\citet{Rubinstein16:Settling} (or even~\citealp*{Babichenko16:Can}) in team zero-sum games?

\section*{Acknowledgments}
Ioannis Panageas is supported by NSF grant CCF- 2454115. 
Tuomas Sandholm is supported by NIH award A240108S001, the Vannevar Bush Faculty Fellowship ONR N00014-23-1-2876, and National Science Foundation grant RI-2312342.

We thank Martino Bernasconi, Matteo Castiglioni, Andrea Celli, and Alexandros Hollender for coordinating the arXiv submissions.

\bibliography{main}

\clearpage
\appendix

\end{document}